\documentclass[11pt]{article}

\pdfoutput=1

\usepackage{hyperref}
\usepackage[svgnames]{xcolor}
\definecolor{darkgreen}{rgb}{0,.35,0}
\definecolor{darkblue}{rgb}{0,0,.5}
\definecolor{darkred}{rgb}{.6,0,0}
\hypersetup{
  colorlinks,
  linkcolor=darkred,
  citecolor=darkblue,
  urlcolor=darkgreen}

\usepackage{amsthm}

\usepackage{amsmath}
\usepackage{amssymb}
\usepackage{pdfpages}
\usepackage{enumitem}
\usepackage{cancel}
\usepackage{hyphenat}
\usepackage{algorithm}
\usepackage{algpseudocode}
\usepackage{algorithmicx}
\usepackage[round, sectionbib]{natbib}
\usepackage{natbib}

\numberwithin{equation}{section}

\newtheorem{theorem}{Theorem}[section]
\newtheorem{corollary}[theorem]{Corollary}
\newtheorem{lemma}[theorem]{Lemma}
\newtheorem{proposition}[theorem]{Proposition}
\newtheorem{definition}[theorem]{Definition}
\newtheorem{fact}[theorem]{Fact}
\newtheorem{example}[theorem]{Example}

\newcommand{\norm}[1]{  \|#1  \|}
\newcommand{\tallnorm}[1]{  \left \|#1 \right \|}

\newcommand{\D}{\partial}

\newcommand{\NN}{{\mathbb{N}}}
\newcommand{\QQ}{{\mathbb{Q}}}
\newcommand{\RR}{{\mathbb{R}}}

\newcommand{\ZZ}{{\mathbb{Z}}}

\newcommand{\degD}{{\deg_{\D}}}

\newcommand{\calF}{{\cal F}}
\newcommand{\calH}{{\cal H}}
\newcommand{\calM}{{\cal M}}
\newcommand{\calS}{{\cal S}}
\newcommand{\Gtil}{{\widetilde G}}
\newcommand{\ftil}{{\widetilde f}}
\newcommand{\gtil}{{\widetilde g}}
\newcommand{\htil}{{\widetilde h}}
\newcommand{\fhat}{{\widehat f}}
\newcommand{\ghat}{{\widehat g}}
\newcommand{\calC}{{\mathcal C}}

\newcommand{\Shat}{{\widehat S}}

\newcommand{\what}{{\widehat w}}

\newcommand{\Sigmabar}{{\overline\Sigma}}

\newcommand{\fstar}{ f^{*}}
\newcommand{\gstar}{g^{*}}
\newcommand{\gstarhat}{ \widehat{ \gstar}}
\newcommand{\fstarhat}{ \widehat{ \fstar}}
\newcommand{\hhat}{ \widehat{ h}}

\newcommand{\false}{{\mbox{\upshape false}}}
\newcommand{\true}{{\mbox{\upshape true}}}
\newcommand{\lclm}{{\mbox{\upshape lclm}}}

\newcommand{\gcrd}{{\mbox{\upshape gcrd}}}

\newcommand{\epgcd}{{\varepsilon{\text{-}}\mbox{\upshape GCD}}}
\newcommand{\tran}{T}
\newcommand{\pvec}[1]{ {\bf {#1} }}

\newcommand{\rconv}[1]{ \calC^{{\bf R}}_{#1} }
\newcommand{\lconv}[1]{ \calC^{ {\bf L}}_{#1} }

\newcommand{\dvec}{{  \overrightarrow{\deg}}}

\renewcommand{\epsilon}{\varepsilon}


\DeclareMathOperator{\lnullspace}{{\mbox{\upshape null}_{\ell}}}

\DeclareMathOperator{\lcoeff}{{\mbox{\upshape lcoeff}}}
\DeclareMathOperator{\content}{{\mbox{\upshape cont}}}

\title{Computing Approximate Greatest Common Right Divisors of Differential
  Polynomials}

\author{
  \parbox{6.7cm}{
    \centering
    Mark Giesbrecht \hspace*{7pt} Joseph Haraldson\\
    \small
    Cheriton School of Computer Science\\
    University of Waterloo\\
    \href{mailto:mwg@uwaterloo.ca} {\{mwg,jharalds\}@uwaterloo.ca}
  } 
  \parbox{5cm}{
    \centering
    Erich Kaltofen\\
    \small
    Dept. of Mathematics\\
    North Carolina State University\\
    \href{mailto:kaltofen@math.ncsu.edu}{kaltofen@math.ncsu.edu}
  }
} 

\begin{document}

\maketitle

\begin{abstract}
  Differential (Ore) type polynomials with ``approximate'' polynomial
  coefficients are introduced.  These provide an effective notion
  of approximate differential operators, with a strong algebraic
  structure. We introduce the approximate Greatest Common Right Divisor
  Problem (GCRD) of differential polynomials, as a non-commutative
  generalization of the well-studied approximate GCD problem.

  Given two differential polynomials, we present an algorithm to find
  \emph{nearby} differential polynomials with a non-trivial GCRD, where
  \emph{nearby} is defined with respect to a suitable coefficient norm.
  Intuitively, given two linear differential polynomials as input, the
  (approximate) GCRD problem corresponds to finding the (approximate)
  differential polynomial whose solution space is the intersection of the
  solution spaces of the two inputs.

  The approximate GCRD problem is proven to be locally well-posed.  A method
  based on the singular value decomposition of a differential Sylvester matrix
  is developed to produce an initial approximation of the GCRD.  With a
  sufficiently good initial approximation, Newton iteration is shown to converge
  quadratically to an optimal solution. Finally, sufficient conditions for
  existence of a solution to the global problem are presented along with
  examples demonstrating that no solution exists when these conditions are not
  satisfied.

\end{abstract}

%
%
%




\section{Introduction}
\label{sec:intro}

The problem of computing the GCRD in a symbolic and exact setting dates back to
\citet{Ore33}, who presents a Euclidean-like algorithm.  See
\citep{BroPet94} for an elaboration of this approach.  \citet{Li97} introduces a
differential-resultant-based algorithm which makes computation of the GCRD very
efficient using modular arithmetic.  The technique of \citet{Li97} is
an extension of ideas presented by \citet{grigorev90} for computing GCRDs of
differential operators.

The analogous approximate GCD problem for usual (commutative)
polynomials has been a key topic of research in symbolic-numeric computing since
its inception.  A full survey is not possible here, but we note the deep
connection between our current work and that of \citep{CGTW95}; see also
\citep{KarLak96}, \citep{SasSas97}, and \citep{ZenDay04}.  Also important to
this current work is the use of so-called structured (numerical) matrix methods
for approximate GCD, such as structured total least squares (STLS) and
structured total least norm (STLN); see \citep{BGM05} and \citep{KYZ05}.
More directly employed later in this paper is the multiple polynomial
approximate GCD method of \citet{KYZ06}.  This latter paper also provides a nice
survey of the current state of the art in approximate GCDs.  Finally, we modify
the proof of \cite{KYZ07b}, an optimization approach to computing the GCD of
multiple, multivariate \emph{commutative} polynomials, to prove the existence of
a globally nearest GCRD.


The goal of this paper is to devise an efficient, numerically robust algorithm
to compute the GCRD when the coefficients in $\RR$ are given approximately.
Given $f,g \in \RR(t)[\D;']$, we wish to find
$\ftil,\gtil\in\RR(t)[\D;']$, where $\ftil$ is \emph{ near } $f$ and $\gtil$ is
\emph{ near } $g$, such that $\degD \gcrd(\ftil,\gtil)\geq 1$, where \emph{near}
is taken with respect to a distributed Euclidean norm.  That is, $\ftil$ and
$\gtil$ have an exact, non-trivial GCRD.

Linear differential polynomials and GCRD's are key tools in finding closed form
symbolic solution of systems of linear differential equations in modern computer
algebra systems like Maple and Mathematica (see, e.g., \citep{SalZim94} and
\citep{AbrLeLi05}).  Equations with real (floating point) coefficients or
parameters are regularly encountered and it is important to understand the
stability of this fundamental tool in this case.  Moreover, floating arithmetic
is potentially much faster than managing large rational coefficients. We regard
this paper as a positive and important initial exploration of this topic.

We commence with necessary preliminaries and well-known results that we expand
upon in the remainder of this introductory section.  In
Section~\ref{sec:lin-alg} we describe a linear algebra formulation of the
approximate GCRD problem and that can be used in conjunction with truncated SVD
\citep{GiesbrechtHaraldson14,Haraldson15,CGTW95} to compute nearby polynomials
with an exact GCRD. Section~\ref{sec:opt-gcrd} reformulates the approximate GCRD
problem as a continuous unconstrained optimization problem. Sufficient
conditions for existence of a solution are provided with an example showing that
when this sufficient condition is not satisfied there is no solution. These
results are complemented by showing that the Jacobian of the residuals has full
rank and under ideal circumstances Newton iteration will converge quadratically.
We generalize some results of \cite{ZenDay04} and \cite{Zeng11} to a
non-commutative Euclidean domain showing that the problem is locally well-posed.
In Section~\ref{sec:implementation} we present our algorithms explicitly,
discuss their complexity and evaluate the numerical robustness of our
implementation on examples of interest.

A part of this work, presenting the SVD-based approach to approximate GCRD, but
without the proof of existence a nearest solution or analysis of the
corresponding optimization, is presented in the workshop paper
\citep{GiesbrechtHaraldson14}.  This is described in Section
\ref{sec:algs-for-gcrd} of this current work.%
\subsection{Preliminaries}
We review some well known results \citep{Ore33} and \citep{BroPet96} on
differential polynomials.

The ring of differential (Ore) polynomials $\RR(t)[\D;']$ over the real numbers
$\RR$ provides a (non-commutative) polynomial ring structure to the linear
ordinary differential operators.  Differential polynomials have found great
utility in symbolic computation, as they allow us to apply algebraic tools to
the simplification and solution of linear differential equations; see
\citep{BroPet94} for a nice introduction to the mathematical and computational
aspects.

Let $\RR(t)[\D;']$ be the ring of differential polynomials over the function
field $\RR(t)$. $\RR(t)[\D;']$ is the ring of polynomials in $\D$ with
coefficients from the commutative field of rational functions, under the usual
polynomial addition along with the non-commutative multiplication defined by
\[
\D y(t) = y(t) \D + y'(t) \text{ for } y(t) \in \RR(t).
\]
Here $y'(t)$ is the usual derivative of $y(t)$ with respect to $t$.

There is a natural action of $\RR(t)[\D;']$ on the space $\cal{C}^\infty[\RR]$
of infinitely differentiable functions $y(t):\RR\to \RR$. In particular, for any
$y(t) \in \cal{C}^\infty[\RR]$,
\[
f(\D) = \sum_{0\leq i \leq M} f_i(t)\D^i \text{ acts on $y(t)$ as } \sum_{0\leq
  i \leq M} f_i(t) \frac{d^i}{dt^i}y(t).
\]
We maintain a right canonical form for all $f\in \RR(t)[\D;']$ by writing
\begin{equation}
  \label{eq:fcanon}
  f = \frac{1}{f_{-1}}\sum_{0\leq i \leq M} f_i \D^i,
\end{equation}
for polynomials $f_{-1},f_0,\ldots, f_M \in \RR[t].$ That is, with coefficients
in $\RR(t)$ always written to the left of powers of $\D$.
An analogous left canonical form exists as well.

A primary benefit of viewing differential operators in this way is that they
have the structure of a left (and right) Euclidean domain.  In particular, for
any two polynomials $f,g \in \RR(t)[\D;']$, there is a unique polynomial
$h\in\RR(t)[\D;']$ of maximal degree in $\D$ such that $f=\fstar h$ and
$g= \gstar h$ for $\fstar,\gstar \in\RR(t)[\D;']$ (i.e., $h$ divides $f$ and $g$
exactly on the right).  This polynomial $h$ is called the Greatest Common Right
Divisor (GCRD) of $f$ and $g$ and it is unique up to multiplication by a unit
(non-zero element) of $\RR(t)$ (we could make this GCRD have leading coefficient
1, but this would introduce denominators from $\RR[t]$, as well as potential
numerical instability, as we shall see).
An important geometric interpretation of GCRDs is that the GCRD $h$ of
differential polynomials $f$ and $g$ is a differential polynomial whose solution
space is the intersection of the solution spaces of $f$ and~$g$.


Approximations require a norm, so we need a proper definition of the {\emph
  {norm}} of a differential polynomial.
\begin{definition}
  We define the \emph{Euclidean norm} for polynomials and a \emph{distributed
    coefficient norm} for differential polynomials as follows:
  \begin{itemize}
  \item[1.] For $p=\sum_{0\leq i\leq d} p_it^i\in\RR[t]$, define
    \[ \norm{p}=\norm{p}_2 = \left (\sum_{0\leq i\leq d} p_i^2\right )^{1/2}. \]
  \item[2.] For $f=\sum_{0\leq i \leq M} f_i \D^i \in \RR[t][\D;']$, define
    \[\norm{f}=\norm{f}_2 = \left (\sum_{0\leq i\leq M} \norm{f_i}_2^2
    \right ) ^{1/2}.\]
  \end{itemize}
\end{definition}

We could extend the above definition of norm over $\RR(t)$ and $\RR(t)[\D;']$.
However it turns out that this is unnecessary and somewhat complicating. In
practice, we perform most of our computations over $\RR[t]$. In the cases where
we are unable to avoid working over $\RR(t)$, we simply solve an associate
problem. This is done by clearing denominators and performing intermediate
computations over $\RR[t]$, then  converting back to the representation over
$\RR(t)$.  Note that the algebraic problem is always computing GCRDs and
cofactors in $\RR(t)[\D;']$, and not the more intricate algebraic domain
$\RR[t][\D;']$; see the discussion below.

\begin{definition}
  For any matrix $S\in\RR[t]^{(M+N) \times (M+N)}$, we define the Frobenius norm
  $\norm{S}_F$ by
  \[
  \norm{S}_F^2 = \sum_{ij} \norm{S_{ij}}^2.
  \]
\end{definition}

%
%
%
\medskip\noindent
\textbf{Main Problem: Approximate GCRD.}~
  Given $f,g\in\RR[t][\D;']$ such that $\gcrd(f,g) =1$ we wish to compute
  $\ftil, \gtil \in \RR[t][\D;']$ with the same coefficient degree
structure\footnote{The polynomial coefficients  of $\D^i$ have the same degree,
i.e. $\deg \ftil_i \leq  \deg f_i$ and $\deg \gtil_i \leq g_i$.} as
$f$ and $g$ such that $h=\gcrd(\ftil,\gtil) $ with
  $D=\degD h \geq 1$ and
  \begin{itemize}
  \item [(i)] $\norm{f-\ftil}^2_2 + \norm{g-\gtil}^2_2 = \varepsilon$  is
		minimized, and
  \item [(ii)] $D$ is the largest possible for the computed distance
		$\varepsilon$.
  \end{itemize}
  The differential polynomial $h$ is said to be an \emph{approximate GCRD} of $f$ and
  $g$ if these conditions are satisfied. In general it is not easy to
  minimize $\varepsilon$, so instead we take a local optimization
approach and compute an upper bound on this quantity. These upper-bounds will
agree with the global minimum if $\varepsilon$ is sufficiently small. The
algorithmic considerations will
generally assume $D$ is fixed without loss of generality, since we can vary $D$
from 1 to $\min\{M,N\}$ to determine the (local) optimal value.
\bigskip

The approximate GCRD problem is a generalization of computing an $\epgcd$
\citep{Schonhage85,CGTW95,KarLak96,EmiGalLom97} in the commutative case.  The
requirement that the GCRD has maximal degree is difficult to certify outside the
exact setting, however this usually is not a problem in our experiments. We
prove that our formulation of the approximate GCRD problem has a solution with a
minimal $\varepsilon$ (opposed to an infimum). Furthermore, if $D$ is fixed,
then for a computed pair of nearby differential polynomials, we are able to
certify that $\varepsilon$ is reasonably close to the optimal value through a
condition number.

In our approach to the approximate GCRD problem we devise methods of performing
division and computing an exact GCRD numerically. These tools are used in
conjunction with our algorithm for computing a nearby pair of differential
polynomials with an exact GCRD via the SVD. The nearby differential polynomials
with an exact GCRD are used as an initial guess in a post-refinement Newton
iteration.

It will also be necessary to define a partial ordering on differential
polynomials. In later sections we will need to make use of this partial ordering
to preserve structure.
\begin{definition}\label{defn:degree-vector}
  Let $\dvec:\RR[t][\D;']\to \ZZ^{M+1}$ be the degree vector function defined as
  \[\dvec(f) = (\deg_t f_0, \deg_t f_1, \ldots, \deg_t f_M), \text{ for }
  f_0,\ldots, f_{M} \in \RR[t].\]
  For $f,g \in \RR[t][\D;']$ with $\degD f = \degD g=M$ we write
  \[ \dvec(f) < \dvec(g) \text{ if } \deg_t f_i \leq \deg_t g_i \text{ for }
  0\leq i \leq M.\]
  We define $\dvec(f)=\dvec(g), \dvec(f)<\dvec(g), \dvec(f)\geq \dvec(g)$ and
  $ \dvec(f) > \dvec(g)$ analogously.
\end{definition}

We note that differential polynomials are written in a canonical ordering with
highest degree coefficients appearing to the left in our examples.  The degree
vector function and most linearizations will appear in reverse order as a
result. For convenience, we will assume that
$\deg 0 = -\infty$.


\begin{definition}
  Let $f \in \RR[t][\D;']$ where $\degD f = M$ is in standard form. The \emph{content}
  of $f$ is given by $\content(f) = \gcd(f_0,f_1,\ldots,f_M)$.  If
  $\content(f) = 1$, we say that the differential polynomial is \emph{primitive}.
\end{definition}
%

\begin{proposition}
  \label{prop:diff-poly-basics}
  The ring $\RR(t)[\D;']$ is a non-commutative principal left (and right) ideal
  domain.  For $f,g \in \RR(t)[\D;']$, with $\degD f=M$ and $\degD g=N$, we have
  the following properties \citep{Ore33}.

  \begin{enumerate}
  \item[(i)] $\degD (fg) = \degD f + \degD g$ ,
    $\degD(f+g) \leq \max \{ \degD f, \degD g\}$.
  \item[(ii)] There exist unique $q,r\in \RR(t)[\D;']$ with $\degD r < \degD g$ such
    that $f=qg+r$ (right division with remainder).
  \item[(iii)] There exists $h \in \RR(t)[\D,']$ of maximal degree in $\D$ with
    $f=\fstar h$ and $g=\gstar h$.  $h$ is called the GCRD (Greatest Common
    Right Divisor) of $f$ and $g$, written $\gcrd(f,g)=h$.  $\fstar$ and
    $\gstar$ are called the left co-factors of $f$ and $g$. The GCRD is unique
    up to multiplication from a unit belonging to $\RR(t)$.
  \item[(iv)] There exist $\sigma,\tau \in \RR(t)[\D;']$ such that
    $\sigma f= \tau g = \ell$ for $\ell$ of minimal degree. $\ell$ is called the
    LCLM (Least Common Left Multiple) of $f$ and $g$, written
    $\lclm(f,g) = \ell$.  The LCLM is unique up to multiplication from a unit
    belonging to $\RR(t)$.
  \item[(v)] $\degD \lclm(f,g) = \degD f + \degD g - \degD \gcrd(f,g)$.
  \end{enumerate}
\end{proposition}

In an algebraic context we can clear denominators of our inputs and assume
without loss of generality that our GCRD belongs to $\RR[t][\D;']$. We will also
assume our inputs and output are primitive.  Again, this is not algebraically
necessary but will be important for the convergence of our subsequent
optimization formulation (see Section \ref{sec:newtconv}).
It is important to note that the co-factors of the GCRD need not belong to
$\RR[t][\D;']$ even if we have $f,g,h \in \RR[t][\D;']$ such that
$\gcrd(f,g)=h$.  This is not unexpected, as a similar situation occurs when
computing GCD's over $\ZZ[x]$, where cofactors in the GCD of primitive
polynomials may well lie in $\QQ[x]\setminus\ZZ[x]$.  In essence, this is a
computational technique to narrow the input domain, not a change to the problem
being considered.

A related but considerably more difficult problem is computing ideal bases and
factorizations completely within $\RR[t][\D;']$.  This has been dealt with
algebraically and in terms of exact computation by a number of authors, though
not with respect to approximate coefficients; see for example
\citep{HeiLev16,GieHeiLev16,BelHeiLev17}.



Most of our results involve transforming a representation of $f\in \RR(t)[\D;']$
into a representation over $\RR(t)^{1\times K}$ for $K \geq \degD f$.  We make
extensive use of the following map.

\begin{definition}
  For $f\in \RR(t)[\D;']$ of degree $M$ in $\D$ as in \eqref{eq:fcanon}, and
  $K>M$, we define
  \[
  \Psi_{K}(f) = \frac{1}{f_{-1}} (f_0,
  f_1,\ldots,f_M,0,\ldots,0)\in\RR(t)^{1\times K}.
  \]
  That is, $\Psi_{K}$ maps polynomials in $\RR(t)[\D;']$ of degree (in $\D$)
  less than $K$ into $\RR(t)^{1\times K}$.
\end{definition}

It will be useful to linearize (differential) polynomials, that is, express them
as an element of Euclidean space.  For $p\in \RR[t]$ with $\deg_tp = d$ we write
\[
\pvec p = (p_0,p_1,\ldots, p_d) \in \RR^{1\times (d+1)}
\]
For $f=f_0+f_1\D+\cdots f_M\D^M\in\RR[t][\D;']$ with $\degD f=M$ and
$\deg_t f_i=d_i$ we write
\[
\pvec{f} = (\pvec {f_0},\ldots, \pvec {f_M}) \in \RR^{1\times L},
\]
where $L=(d_0+1)+\cdots+(d_M+1)$.  If $d\geq \max\{d_i\}$ we will sometimes pad
each $\pvec{f_i}$ with zeros to have precisely $d+1$ coefficients, and by a
slight abuse of notation regard
\[
\pvec{f}\in\RR^{1\times (M+1)(d+1)}.
\]
We will not do this unless specifically stated.

%
%
%


\section{Computing the GCRD via Linear Algebra}
\label{sec:lin-alg}
\label{ssec:diff-sylvester}

In this section we demonstrate how to reduce the computation of the GCRD to that
of linear algebra over $\RR(t)$, and then over $\RR$ itself.  This approach has
been used in the exact computation of GCRDs \citep{Li97} and differential
Hermite forms \citep{GieKim13}, and has the benefit of reducing differential,
and more general Ore problems, to a system of equations over a commutative
field.  Here we will show that it makes our approximate version of the GCRD
problem amenable to numerical techniques.  We note that for computing
approximate GCRDs of differential polynomials, much as for computing approximate
GCDs of standard commutative polynomials, the Euclidean algorithm is numerically
unstable, and thus we employ resultant-based techniques, as described below.

Since $\RR(t)[\D;']$ is a right (and left) Euclidean domain \citep{Ore33}, a GCRD
may be computed by solving a Diophantine equation corresponding to the B\'ezout
coefficients.  Using the subresultant techniques of \citet{Li98}, we are able to
transform the non-commutative problem over $\RR(t)[\D;']$ into a commutative
linear algebra problem over $\RR(t)$. This is done through a Sylvester-like
resultant matrix. By using resultant-like matrices we are able to express the
B\'ezout coefficients as a linear system over $\RR(t)$ and compute a GCRD via
nullspace basis computation.
%


\begin{lemma}
  \label{lem:nontrivgcrd}
  Suppose $f,g \in \RR(t)[\D;']$ with $\degD f=N$ and $\degD g=M$.  Then
  $\degD \gcrd(f,g) \geq 1$ if and only if there exist $u,v \in \RR(t)[\D;']$
  such that $\degD u< M$, $\degD v< N$, and $uf+vg = 0$.
\end{lemma}
\begin{proof}
  This follows immediately from Proposition \ref{prop:diff-poly-basics}.
  \qed
\end{proof}

Using Lemma~\ref{lem:nontrivgcrd} we can solve a B\'ezout-like system to compute
a GCRD of two differential polynomials. This is characterized by the
differential Sylvester matrix, based on the subresultant method of \citet{Li97}.

\begin{definition}
Suppose $h\in\RR[t][\D;']$ has $\degD=D$.    For any $K\in\NN$, the matrix
  \[
  \rconv{K}(h) =
  \begin{pmatrix}
    \Psi_{K+D+1}(h) \\
    \Psi_{K+D+1}(\D h) \\
    \vdots \\
    \Psi_{K+D+1}(\D^{K} h)
  \end{pmatrix}
  \in \RR[t]^{(K+1) \times (K+D+1)}
  \]
  is the $K^{th}$ \emph{right differential convolution matrix} of $h$.  We note
  that the entries of $\rconv{K}(h)$ are written in their {\emph{right canonical
      form}}, where the $\D$'s appear to the right (polynomials in
$\RR[t]$ appear to the
  left).  We note that $\deg_t \D^i h = \deg_t h$, so the degree in $t$ of all
  entries of $\rconv{K}(h)$ is at most $\deg_t h$.

  We analogously define the $K^{th}$ \emph{left differential convolution matrix}
  of $h$ as $\lconv{K}(h)$ as
  \[
  \lconv{K}(h) = \begin{pmatrix}
    \Psi_{K+D+1}(h) \\
    \Psi_{K+D+1}( h\D) \\
    \vdots \\
    \Psi_{K+D+1}(h \D^{K} )
  \end{pmatrix} \in \RR[t]^{(K+1) \times (K+D+1)},
  \]
  where elements are written in their {\emph{left canonical form}}, where the
  $\D$'s appear to the left (polynomials in $\RR[t]$ always appear to the
right).
\end{definition}

Both right and left differential convolution matrices can be used to perform
multiplication.  Suppose $f^*\in\RR(t)[\D;']$, $h\in\RR(t)[\D;']$ and
$f=\fstar h\in\RR[t][\D;']$, with
\begin{equation}
\label{eq:fhnum}
f = \sum_{0\leq i \leq M} f_i \D^i, \; \fstar = \sum_{0\leq i \leq M-D} \fstar_i
\D^i \text{ and } h = \sum_{0\leq i \leq D} h_i \D^i,
\end{equation}
with $f_i,h_i\in \RR[t]$ and $\fstar_i \in \RR(t)$.  We can express the product
of $\fstar$ and $h$ as
\[
(f_0,f_1,\ldots, f_{M}) = (\fstar_0,\ldots, \fstar_{M-D}) \rconv{M-D}(h).
\]
Similarly, we may write
\[
(f_0,f_1,\ldots, f_M)^\tran = \lconv{D}(\fstar) (h_0,h_1,\ldots,h_D)^\tran.
\]
In keeping with our canonical ordering, we express our results in terms of right
differential convolution matrices.  We carefully observe that both the right and
left differential convolution matrices described correspond to {\emph{right
    multiplication}}. Left multiplication can be formulated in a similar manner.

Let $f,g \in \RR(t)[\D;']$ with $\degD f = M$ and $\degD g =N$. Then by
Lemma~\ref{lem:nontrivgcrd} we have that $\degD \gcrd(f,g) \geq 1$ if and only
if
there exist $u,v \in \RR(t)[\D;']$ such that $\degD u < N, \degD v < M$ and
$uf+vg =0$.  We can encode the existence of $u,v $ as an $(M+N)\times (M+N)$
matrix over $\RR(t)$ in what we will call the differential Sylvester matrix.

\begin{definition}
  The matrix
  \[
  S = S(f,g) =
  \begin{pmatrix}
    \rconv{N-1}(f) \\
    \rconv{M-1}(g)
  \end{pmatrix}
  \in \RR(t)^{(M+N) \times (M+N)}
  \]
  is the differential Sylvester matrix of $f$ and $g$.
\end{definition}

This matrix \citep{Li97} is analogous to the Sylvester matrix of real
polynomials; see \cite[Chapter~6]{MCA3}. As expected, many useful properties
of
the Sylvester matrix over real polynomials still hold with the differential
Sylvester matrix.  These similarities become evident when we consider
\[
w = (u_0,u_1,\ldots, u_{N-1},v_0,v_1,\ldots, v_{M-1}) \in \RR(t)^{1\times
  (M+N)}.
\]
Then $uf+vg=0$ implies that $wS=0$, hence $w$ is a non-trivial vector in the
(left) nullspace of $S$. In particular, this solution is equivalent to saying
that $S$ is singular.  Clearing denominators of $f$ and $g$, we may assume that
$u,v \in \RR[t][\D;']$, i.e., they have polynomial coefficients, which implies
that $S\in \RR[t]^{ (M+N) \times (M+N)}$.  Moreover, for $f,g \in \RR[t][\D;']$
with $\deg_t f \leq d$ and $\deg_t g \leq d$ then $\deg_t S_{ij} \leq d$.

We summarize these results in the following lemma.
\begin{lemma}
  \label{lem:RRtGCRD}
  Suppose $f,g\in\RR[t][\D;']$, where $\degD f=M$, $\degD g=N$, $\deg_t f\leq d$
  and $\deg_t g\leq d$.
  \begin{enumerate}
  \item[(i)] $S=S(f,g)$ is singular if and only $ \degD \gcrd(f,g)\geq 1$.
  \item[(ii)] $\degD\gcrd(f,g)= \dim \lnullspace(S)$, where $\lnullspace(S)$ is the
    left nullspace of $S$.
  \item[(iii)] For any
    $w = (u_0,\ldots, u_{N-1}, v_0, \ldots, v_{M-1}) \in \RR(t)^{1 \times
      (M+N)}$
    such that $wS=0$, we have $uf+vg=0$, where $u=\sum_{0\leq i<N} u_i\D^i$ and
    $v=\sum_{0\leq i<M} v_i\D^i$.
  \item[(iv)] Suppose that $\degD\gcrd(f,g)\geq 1$.  Then there exists
    $w\in \RR[t]^{1\times (M+N)}$ such that $wS = 0$ and
    $\deg_t w \leq \mu = 2(M+N)d$.
  \end{enumerate}
\end{lemma}
\begin{proof}
  Part (i) -- (iii) follow from Lemma \ref{lem:nontrivgcrd} and the discussion
  above.  Part (iv) follows from an application of Cramer's rule and a bound on
  the degree of the determinants of a polynomial matrix.
  \qed
\end{proof}


%
%
%

\subsection{Linear Algebra over $\RR$}
\label{ssec:inflated-diff-sylvester}
%

Let $S\in \RR[t]^{(M+N) \times (M+N)}$ be the differential Sylvester matrix of
$f, g\in\RR[t][\D;']$ of degrees $M$ and $N$ respectively in $\D$, and degrees
at most $d$ in $t$.  From Lemma \ref{lem:RRtGCRD} we know that if a GCRD of $f$
and $g$ exists, then there is a $w\in\RR[t]^{1\times (M+N)}$ such that $wS=0$,
with $\deg_t w\leq \mu = 2(M+N)d$.

\begin{definition}
  The $k^{th}$ convolution matrix of $b \in \RR[t]$ with $\deg b = m$ is defined
  as
  \[ C_k (b) = \begin{pmatrix}
    b_0	& 	  &\\
    b_1 	& \ddots  &\\
    \vdots 	& \ddots  & b_0\\
    b_m	&	  & b_1\\
    & \ddots  & \vdots \\
    & & b_m
  \end{pmatrix} \in \RR^{(m+k+1) \times (k+1)}.
  \]
  Let $a\in \RR[t]$ with $\deg a =\mu$ and define the mapping
  $\Gamma: \RR[t] \to \RR^{ (\mu+1) \times (\mu+d+1)}$ by
  $\Gamma(a) = C_{d}(a)^\tran $. $\Gamma(a)$ is the left multiplier matrix of
  $a$ with respect to the basis $\langle 1,t,\ldots,t^{\mu+d} \rangle$.
\end{definition}
A differential convolution matrix generalizes the convolution matrix in the role
of linearizing multiplication between differential polynomials.

\begin{definition}
  Given the $(M+N) \times (M+N)$ differential Sylvester matrix $S$, we apply
  $\Gamma$ entry-wise to $S$ to obtain
  $\Shat\in \RR^{(M+N)(\mu+1) \times (M+N)(\mu +d +1)}$; each entry of $S$ in
  $\RR[t]$ is mapped to a block entry in $\RR^{(\mu+1)\times (\mu+d+1)}$ of
  $\Shat$.  We refer to $\Shat$ as the \emph{inflated differential Sylvester
    matrix} of $f$ and $g$.
\end{definition}


%

\begin{lemma}
  \label{lem:RRGCRD}
  Let $f,g\in\RR[t][\D;']$ have differential Sylvester matrix
  $S\in\RR[t]^{(M+N)\times (M+N)}$ and inflated differential Sylvester matrix
  \[
  \Shat\in\RR^{(M+N)(\mu+1)\times (M+N)(\mu+d+1)}.
  \]
  There exists a
  $w\in \RR[t]^{1\times (M+N)}$ such that $wS=0$, if and only if there exists a
  $\what \in \RR^{(\mu+d+1)\times (M+N)(\mu+1)}$ such that $\what \Shat =0$.
  More generally,
  \[
  \degD\gcrd(f,g)= \frac{\dim\lnullspace(\Shat)}{\mu+d+1}.
  \]
\end{lemma}
\begin{proof}
  This follows directly from the definition of $\Gamma$ and Lemma
  \ref{lem:RRtGCRD}.
  \qed
\end{proof}

We note that $\Shat$ is no longer a square matrix. This will not pose too many
problems as we will see in the following sections.

\label{sec:approxgcrd}
%
%
%
%

\subsection{Division Without Remainder }
\label{ssec:diff-division}
While multiplication of differential polynomials with approximate numerical
coefficients is straightforward, division is somewhat more difficult.  We will
generally require a division \emph{without remainder}, for the computation of
which we use a least squares approach.  Given $f,h\in\RR[t][\D;']$ as in
\eqref{eq:fhnum}, we wish to find an $\fstar\in\RR(t)[\D;']$ such that
$\norm{f-f^*h}$ is minimized.  We will assume as usual that $\degD f=M$,
$\degD h=D$ and $\deg_t f, \deg_t h\leq d$.


Much as in the (approximate polynomial) commutative case, we do this by setting
the problem up as a linear system and then finding a least squares solution.
Let us assume for now that $f= \fstar h$ is exact, so this can be expressed as a
linear system over $\RR(t)$ by writing
\begin{equation}
\label{eqn:co-factor-system}
(f_0,f_1,\ldots, f_{M}) =
(\fstar_0,\ldots, \fstar_{M-D}) \rconv{M-D} (h).
\end{equation}

This system of equations is over-constrained (over $\RR(t)$), but we note that
the sub-matrix formed from the last $M-D+1$ columns of $\rconv{M-D}(h)$ is lower
triangular, with diagonal entry $h_D\in\RR[t]$.  Thus, any exact quotient
$h\in\RR[t][\D;']$ such that $f=\fstar h$, in lowest terms, must have
denominators dividing $h_D^{M-D+1}$, and in particular have denominators of
degree at most $ (M-D+1) \deg_t h_D  \leq (M-D+1)d$.  Equivalently,
$h_D^{M-D+1} \fstar \in\RR[t]^{M-D+1}$.  By applying Cramer's rule on the last $M-D+1$ columns of $\rconv{M-D}(h)$, the degrees of the
numerators in $\fstar$ must be at most $(M-D+1)d$.  Using this information we
can formulate an associated problem with coefficients from $\RR[t]$ and avoid
performing linear algebra over $\RR(t)$.

Now let $v_{-1},v_0,\ldots,v_{M-D}$ be generic polynomials in $t$, with
indeterminate coefficients of degree at most $(M-D+1)d$.   I.e.,
\[
v_i = \sum_{j=0}^{(M-D+1)d} v_{ij}t^j, ~~~\mbox{$i=-1\ldots {M-D}$,}
\]
for indeterminates $v_{ij}$ with $v_{-1}\neq 0$.  Then we are seeking to solve the linear system of equations
\[
v_{-1} \cdot (f_0,\ldots,f_M) = (v_0,\ldots,v_{M-D}) \ \rconv{M-D}(h)
\]
for the $v_{ij}$.  For each entry $f_i$ we have $(M-D+1)d+d+1$ equations; this
is the degree $(v_0,\ldots,v_{M-D}) \rconv{M-D}(h)$ plus one, and we get one
equation per coefficient.  Hence there are $(M+1)((M-D+1)d+d+1)$ equations in
$(M-D+2)(M-D+1)d$ unknowns.  We then use a standard linear least squares
solution to find the $v_{i}$ which minimizes the residual, and thus minimizes
$\norm{f-\fstar h}$.

It may be desirable to find the lowest degree $v_{-1}$ which meets this
criteria, for which we can use a simple binary search for a lower degree with
reasonable residual (or alternatively use an SVD-based identification procedure).

Finally, a more straightforward approach to solving \eqref{eqn:co-factor-system}
is to simply use the solution from the last $M-D+1$ columns of
$\rconv{M-D}(h)$. The last $M-D+1$ columns of $\rconv{M-D}(h)$ are lower
triangular, with diagonal entries consisting of $h_D\in\RR[t]$.  While this does
not yield a solution to the least squares normal equations, it is usually
sufficiently good in practice, and considerably easier to formulate.

\section{Optimization-based Formulation of Approximate GCRD}
\label{sec:opt-gcrd}
First we standardize some notation and assumptions.  We assume that
$f,g,\ftil,\gtil,h \in \RR[t][\D;']$ and $\fstar,\gstar\in\RR(t)[\D;']$.
Moreover, we assume that $\ftil=\fstar h$ and $\gtil = \gstar h$ and
$h=\gcrd(\ftil,\gtil)$.  Intuitively, $f,g$ are our ``input polynomials'' and we
will be identifying ``nearby'' $\ftil,\gtil$ with a non-trivial GCRD $h$.  Note
that $\fstar,\gstar$ have rational function coefficients.
Later we will find it useful to clear fractions and work with a primitive
associate.

We also assume degree bounds as follows: $\degD f =\degD\ftil = M$, \linebreak
$\deg_t f, \deg_t\ftil \leq d$, $\degD g, \degD \gtil = N$,
$\deg_t g,\deg_t\gtil \leq d$, $\degD h = D$, $\degD \fstar = M-D$ and
$\degD \gstar = N-D$.

Using the method of \cite{GiesbrechtHaraldson14}, essentially the generalization
of the SVD-based method of \cite{CGTW95} to differential polynomials, we will
make an initial guess for $\ftil,\gtil$;  details are described in Section
\ref{sec:algs-for-gcrd} of this paper.  We then use optimization techniques to
hone in on polynomials with minimal distance.  While the techniques in that
paper are not particularly effective at providing a nearest solution, they do
provide a suitable initial guess, which we employ here.

We next describe how to formulate an objective function $\Phi$ that, when
minimized, corresponds to a solution to the approximate GCRD problem.
Define the objective function $\Phi:\RR[t][\D;']\times \RR(t)[\D;']^2 \to \RR$
as
\[
\Phi(h,\fstar, \gstar) = \norm{f-\fstar h}_2^2 + \norm{g-\gstar h}_2^2.
\]
In keeping up with our notation from earlier, we observe that $\ftil = \fstar h$
and $\gtil = \gstar h$ in the context of the objective function $\Phi$, as $f$
and $g$ will typically be relatively prime.  To compute guesses for the
co-factors given $h$, we will perform an approximate division without remainder
using the method of Section \ref{ssec:diff-division}.  We only require an
initial guess for $\fstar$ and $\gstar$ to minimize $\Phi$, so this
factorization doesn't need to be exact, in the event that $\gcrd(f,g) = h$.

We show that $\Phi$ has an attainable global minimum under appropriate
assumptions. More precisely, there exist non trivial $\ftil$ and $\gtil$ such
that
\begin{equation}
  \label{eq:ftilgtil}
  \norm{f-\ftil}_2^2 + \norm{g-\gtil}_2^2
\end{equation}
is minimized. Furthermore, we will show that the approximate GCRD problem is
locally well-posed.
%
%
%

\subsection{Existence of Solutions}

\begin{lemma}
  \label{lem:co-factor-bound}
  Let $f, h \in \RR[t][\D;']$, with monic leading coefficients, be not
  necessarily primitive, such that $f=\fstar h$ for $\fstar \in \RR[t][\D;']$
  with $\degD f = M$ and $\degD h =D$.  Then $\norm{\fstar}$ is bounded above.
\end{lemma}
\begin{proof}
  It follows that $\fstar$ is bounded by the computing the Cramer solution
  to \eqref{eqn:co-factor-system} using the last $M-D+1$ columns of
  $\rconv{M-D}(h)$.
\end{proof}

As an observation, we relax the assumption that $f$ is primitive (we work with
an associate instead) in order to guarantee that $\fstar \in \RR[t][\D;']$. This
can be taken without loss of generality as the quantity $\norm{\content(f)}_2^2$
is bounded above and away from zero (as its leading coefficient is monic). Thus
we may divide by it without affecting the quality of the results, as
$\norm{f-\fstar h}$ is still well defined.

We will make use of the following well known fact from
\citep[Theorem~4.16]{rudin76}.
\begin{fact}
  \label{thm:weirstrass-extreme-value}
  Suppose that $\Phi$ is a continuous real function on a compact metric space
  $X$. Then there exist points $p$ and $q$ in $X$ such that
  \[
  \Phi(p) \leq \Phi(x) \leq \Phi (q),
  \]
  for all $x\in X$. Precisely, $\Phi$ attains its minimum and maximum values at
  $p$ and $q$ respectively.
\end{fact}
%
%
%

We first state a general version of the theorem where a logical predicate
$\Xi:\RR^k\to \{\text{true},\text{false}\}$ (for some $k$) can be chosen to
impose additional constraints on the problem.  For the rest of this section let
\[
\phi: \RR[t][\D;']^2 \to \RR^{(M+N+2)(d+1)}
\]
be the combined coefficient vector function, i.e. for arbitrary
$f,g \in \RR[t][\D;']$ we write $\phi(f,g) = (\pvec f, \pvec g)$, where
$\pvec f$ and $\pvec g$ are padded with zeros to have the desired dimensions.

The following lemma and its proof are analogous to \cite [Theorem~2]{KYZ07b},
which in turn generalizes the univariate argument of \cite [Theorem~1]{KYZ07}.

\begin{theorem}[Existence of Global Minima]
  \label{thm:global-min} Let $f,g \in \RR[t][\D;']\backslash \{0 \}$, let
  $d=\max\{ \deg_t f, \deg_t g \}$, $\degD f=M$, $\degD g=N$ and
  $D\leq \min \{M,N \}$.  Furthermore, let
  $\Xi: \RR^{(M+N+2)(d+1)} \to \{\true, \false \}$ be a predicate on
  $\phi(f,g)$. We assume that the preimage $\Xi^{-1}(\true)$ is a topologically
  closed set in $\RR^{(M+N+2)(d+1)}$ with respect to the Euclidean norm.  For a
  given $\Omega\in \RR_{>0}$ we define the set of possible solutions by
  \[
  \calF_{\Omega} = \left \{
    \begin{array}{cc} (\ftil,\gtil) \in \RR[t][\D;']^2 \text{ such that }
      &  \degD \ftil = M, \\ &  \degD \gtil = N, \\
      & \degD \htil \geq D, \\
      &  \htil = \gcrd(\ftil,\gtil) ,\\ & \norm{\htil} \leq \Omega, \\
      &  \lcoeff_t (\lcoeff_\D  \htil) =1, \\
      &  \text{ and }\Xi( \phi(\ftil,\gtil)) = \true
    \end{array}
  \right \}.
  \]
  Suppose that $\calF_{\Omega} \neq \emptyset$. Then the minimization problem
  \begin{equation}
    \min_{(\ftil,\gtil) \in \calF_\Omega} \norm{f-\ftil}_2^2 + \norm{g-\gtil}_2^2 \label{eq:opt-function}
  \end{equation}
  has an attainable global minimum.
\end{theorem}

\begin{proof}
  Without loss of generality, we assume that $M\leq N$.  Then we iterate the
  minimization over all $\ell \in \ZZ_{\geq 0}$ such that $D\leq \ell \leq M$
  and coefficients  $\rho \subset \RR[t]^{\ell}$. Let
$\calH_{\ell,\rho}$ denote the
  set of all differential polynomials over $\RR[t][\D;']$ of degree $\ell$ with
  coefficients from $\rho$.
  We optimize over the continuous real objective function
  \[
  \Phi(h,\fstar, \gstar) = \norm{f-\fstar h}_2^2 + \norm{g-\gstar h }_2^2,
  \]
  for $h\in \calH_{\ell,\rho}$, $\degD \fstar \leq M-D$ and
  $\degD \gstar \leq N-D$.  We fix the leading coefficient of $h$ with respect
  to $\D$ to be monic, that is $\lcoeff_t \lcoeff_\D h =1$.

  Since the leading coefficient of $h$ is monic, we can write
  $G= \gcrd(\fstar h, \gstar h)$ with $\degD G \geq D$.  Since $G$ is a multiple
  of $h$, we normalize $G$ so that $\lcoeff_t \lcoeff_\D G = 1$, i.e. the
  leading coefficient of $G$ is also monic. The restriction on $h$ that the
  leading coefficient of $h$ is monic enforces that $\degD G \geq D$.
  Furthermore, we restrict the domain of our function $\Phi$ to those $h,\fstar$
  and $\gstar$ for which $(\fstar h, \gstar h) \in \calF_\Omega$.
  If there is no such common factor $h$ and co-factors $\fstar$ and $\gstar$,
  then this pair of $\ell$ and $\rho$ does not occur in the minimization
  \eqref{eq:opt-function}. By assumption we have that
  $\calF_\Omega \neq \emptyset$, so there must be at least one possible case.
  We note that if $(0,0) \in \calF_\Omega$, then $\fstar = \gstar =0$.

  Now suppose that for the given $\ell$ and $\rho$, there are
  $\htil \in \calH_{\ell,\rho}$ and $\ftil^*, \gtil^*$ satisfying
  $\degD \ftil^* \leq M-\ell$ and $\degD \gtil^* \leq N-\ell$ such that
  $(\ftil^* \htil, \gtil^* \htil) \in \calF_\Omega$. We shall prove that the
  function $\Phi$ has a value on a closed and bounded set (i.e., compact with
  respect to the Euclidean metric) that is smaller than elsewhere.  Hence $\Phi$
  attains a global minimum by Fact~\ref{thm:weirstrass-extreme-value}.


  Clearly any solution $\htil \in \calH_{\ell,\rho}$ and $\ftil^*, \gtil^*$ with
  $(\ftil^* h, \gtil^* h) \in \calF_\Omega$ but with
  $\Phi(\htil, \ftil^*, \gtil^*) >
  \Phi(h,\fstar,\gstar)$ 
  can be discarded.  So the norm of the products $\norm{\ftil^*
    \htil}_2$ and $\norm{\gtil^*
    \htil}_2$ can be bounded from above.  We have that
  $\norm{\htil}$
  is bounded above by Lemma \ref{lem:co-factor-bound} because it is a right
  factor of $\Gtil
  = \gcrd(\ftil^* \htil, \gtil^* \htil)$ with $\norm{\Gtil} \leq
  \Omega$. We note that
  $\htil$
  has a monic leading coefficient, so $\norm{\htil}
  \geq 1$.  We have that $\norm{\ftil^*}$ and
  $\norm{\gtil^*}$
  (or the appropriate associate) are both bounded above by Lemma
  \ref{lem:co-factor-bound}.

  Thus we can restrict the domain of $\Phi$
  to values that lie within a sufficiently large closed ball $B$.
  The function $\zeta
  $ that maps $(h,\fstar,
  \gstar)$ to the combined coefficient vector $\phi(\fstar h,\gstar
  h)$ of $\fstar h$ and $\gstar h$ is continuous. We minimize over $\zeta
  ^{-1}(\Xi^{-1}(\true)\cap \zeta (B))$, which is a compact set.
  \qed
\end{proof}
%

For the less general version of the theorem, given arbitrary
$f,g \in \RR[t][\D;']$, we define
\[
\calS = \calS(f,g)=\left \{ \phi(\ftil, \gtil) \; | \; \ftil,\gtil \in
  \RR[t][\D;']
  ~~\textrm{such that}~~
  \parbox{3cm}{
    $\dvec(\ftil) \leq \dvec (f)$\\
    $\dvec (\gtil) \leq \dvec(g)$
   }
  \right\}.
\]
We observe that $\calS$ is a closed subset of $\RR^{(M+N+2)(d+1)}$, where
$\degD f =M$, $\degD g = N$ and $d = \max\{ \deg_t f, \deg_t g\}$.  The set
$\calS$ corresponds to the combined coefficient vectors of $\ftil$ and $\gtil$
that have the same degree structure as $f$ and $g$.


\begin{corollary}
  \label{cor:global-min-spec}
  Let $f,g \in \RR[t][\D;']\backslash \{0 \}$, let
  $d=\max\{ \deg_t f, \deg_t g \}$, $\degD f=M$, $\degD g=N$ and
  $D \leq \min \{M,N\}$.  For a given $\Omega\in \RR_{>0}$ we define the set of
  possible solutions by
  \[
  \calF_{\Omega} = \left \{
    \begin{array}{cc}
      (\ftil,\gtil) \in \RR[t][\D;'] \times \RR[t][\D;'] \text{ such that }
      &  \degD \ftil = M, \\ & \degD \gtil = N, \\
      &  \phi(\ftil,\gtil) \in \calS,\\
      &  \degD \htil \geq D, \\
      &  \htil = \gcrd(\ftil,\gtil) ,\\ & \norm{\htil} \leq \Omega, \\
      &  \hspace*{-13pt}\text{ and } \lcoeff_t (\lcoeff_\D (\htil)) =1\\
      %
    \end{array} \right \}.
  \]
  Suppose that $\calF_{\Omega} \neq \emptyset$. Then the minimization problem
  \begin{equation*}
    \min_{(\ftil,\gtil) \in \calF_\Omega} \norm{f-\ftil}_2^2 + \norm{g-\gtil}_2^2
  \end{equation*}
  has an attainable global minimum.
\end{corollary}


We note that Theorem \ref{thm:global-min} does not guarantee a unique minimum of
$\Phi$, merely that $\Phi$ has an attainable minimum (as opposed to an
infimum). The choice of $h,\fstar$ and $\gstar$ that we optimize over is
important.  If $\lcoeff_t \lcoeff_\D h$ vanishes or $\norm{h_0},\ldots,
\norm{h_{D-1}}$ are quite large, then
$\fstar$ and $\gstar$ can be ill-conditioned in the approximate GCRD problem.
Furthermore, choosing overly
large, small or poor degree structure in $t$ for $h$ can result in a $\Phi$ that
cannot be minimized for the specified structure, but would otherwise have a
minimum for a different choice of $h$.

\begin{example}
  Consider $f=\D^2-2\D+1$ and $g=\D^2+2\D+2$ (see \citep{KYZ07b} for
 an example with complex perturbations). Then $f$ and $g$ do not have a
  degree 1 approximate GCRD. That is, we show that there does not exist
  $\ftil,\gtil \in \RR[t][\D;']$ where $\degD \gcrd(\ftil,\gtil) =1$ and
  $\norm{f-\ftil}_2^2 + \norm{g-\gtil}_2^2$ is minimized.

  The real monic Karmakar-Lakshman distance \citep{KarLak96,KarLak98} of
  \[
  \norm{f-\ftil}_2^2 + \norm{g-\gtil}_2^2
  \]
  occurs when the rational function
  \begin{equation*}
    \frac{2h_0^4+14h_0^2+4h_0 + 5}{h_0^4+h_0^2+1}
  \end{equation*}
  is minimized for $h_0 \in \RR$. The minimum value (if it exists) of this
  function corresponds to the approximate GCRD $h=\D - h_0$.
  The infimum is $2$, which is unattainable. There is no attainable global
  minimum.

  The non-monic real Karmakar-Lakshman distance is $2$, which is achieved if and
  only if the leading coefficient vanishes. The minimum occurs when the
  rational function
  \[
  \frac{5h_1^4 - 4h_1^3+14h_1^2+2}{h_1^4 + h_1^2 +1}
  \]
  is minimized. The minimum value of this function corresponds to the
  approximate GCRD $h = h_1 \D +1$.

  In particular, if we consider $\ftil = (-2\D+1)(\varepsilon \D +1)$ and
  $\gtil = (2\D +2)(\varepsilon \D+1)$, then
  $\norm{f-\ftil}_2^2 + \norm{g-\gtil}_2^2$ becomes arbitrarily near $2$ as
  $\varepsilon \to 0$.

  There is no real degree $1$ approximate GCRD, as
  \[
  \min_{\left \{ (\ftil,\gtil) \in \RR[t][\D;']^2 \; | \; \degD
      \text{\gcrd}(\ftil,\gtil) =1 \right \} } \norm{f-\ftil}_2^2 +
  \norm{g-\gtil}_2^2
  \]
  is not defined in the monic case. In the non-monic case, if a minimum exists
  then it occurs when $\lcoeff_\D h$ vanishes, so the minimum value is not
  defined either.

  This example illustrates that not all $f,g \in \RR[t][\D;']$ have an
  approximate GCRD. Furthermore, we see that the requirement that
  $\lcoeff_t \lcoeff_\D h =1$ and $\norm{h}$ is bounded, from
  Theorem~\ref{thm:global-min} are required, even if there are no additional
  constraints imposed.
\end{example}


Now it remains to show that it is possible to obtain a (locally) unique solution
to $\Phi$.  One of many equivalent conditions for uniqueness of an exact GCRD,
is to require it to be primitive and have a monic leading coefficient.
Numerically, to obtain a unique solution of the approximate GCRD problem, we
impose the same constraints, making solutions locally unique.

\subsection{Convergence of Newton Iteration and Conditioning}
\label{sec:newtconv}

From Theorem \ref{thm:global-min} and Corollary \ref{cor:global-min-spec} we
know a solution to the approximate GCRD problem exists.  We now show that a
standard Newton iteration will converge quadratically when starting with an
estimate sufficiently close to an approximate GCRD.  We first describe the
Jacobian of the residuals and show that the Jacobian has full rank. This leads
to a first-order approximation of the Hessian matrix showing that it is locally
positive definite around a global minimum when the residual is sufficiently
small.  The implication is that Newton's method will converge
quadratically~\citep{Boyd2004}.  If we consider structured perturbations, then
we are able to obtain results similar to that of \cite{ZenDay04} to the overall
conditioning of the system.

In this section we assume without loss of generality that
$\fstar, \gstar \in \RR[t][\D;']$ are primitive, and that $f$ and $g$ may no
longer be primitive to simplify computations.  We need to clear fractions of
rational functions to apply our coefficient norms, and to linearize $h,\fstar$
and $\gstar$ as vectors of real numbers.


The residual of the approximate GCRD is
\begin{align*}
  r &=r(h,\fstar, \gstar)  \\
    &= (\pvec{\fstar h} - \pvec f, \pvec {\gstar h}-\pvec g)^\tran \in \RR^{\eta \times 1},
\end{align*}
where
\begin{align*}
\eta & = \sum_{0 \leq i \leq M} \max\{ \deg_t f_i,-1\} + \sum_{0 \leq i \leq N} \max\{ \deg_t g_i,-1\} + (M+1)+(N+1) \\
     & \leq  (M+N+2)(d+1).
\end{align*}
Intuitively, $\eta $ represents the number of components of
$(\pvec f, \pvec g) \in \RR^{1\times \eta }$.  Let $\nu$ be the number of
variables needed to represent the coefficients of $h,\fstar$ and $\gstar$, i.e.
$(\pvec h, \pvec \fstar, \pvec \gstar) \in \RR^{1\times \nu}$.


Recall that when $f=\fstar h$, we can linearize this relationship with
differential convolution matrices, by writing
\[
f = (\fstar_0,\ldots, \fstar_{M-D}) \rconv{M-D}(h).
\]
If $f_i$ is a coefficient of $f$ with $\deg_t f_i = d$, then we may write
\[
f_i = \sum_{0\leq j \leq M-D} \fstar_j (\rconv{M-D} (h) [j,i]).
\]
This relationship may be linearized over $\RR$ through the use of convolution
matrices.  Writing
\[
\pvec {f_i} = \sum_{0 \leq j \leq M-D} \pvec{\fstar_j} \cdot C_{d} \left(
  \rconv{M-D}(h)[j,i]\right)^\tran,
\]
we now have a direct method of computing $\pvec{f_i}$ in terms of the
coefficients of $\fstar$ and~$h$.

If we differentiate $\pvec{\fstar h}$ with respect to an entry from
$\pvec{\fstar}$, then we will obtain the corresponding (linearized) row of
$\rconv{M-D}(h)$.  Similarly, differentiating $\pvec{ \fstar h}$ with respect to
an entry of $\pvec{h}$ will give us a (linearized) column of
$\lconv{D}(\fstar)$.  This relationship becomes clear when we observe that
\[
(\fstar_0,\ldots, \fstar_{M-D}) \rconv{M-D}(h) = \left( \lconv{D}
  (\fstar) \begin{pmatrix}h_0\\ \vdots \\ h_D \end{pmatrix}\right)^\tran .
\]
Differentiating $\pvec{\gstar h}$ with respect to variables from $\pvec{\gstar}$
and $\pvec{h}$ will produce similar results.


%

The Jacobian of $r(h,\fstar,\gstar)$ for arbitrary $h,\fstar$ and $\gstar$ may
be expressed (up to column permutation) in block matrix form as
\[
J= \begin{pmatrix}
  \rconv{M-D} (h)^\tran	 & 	0 		& \lconv{D} (\fstar) \\
  0 &\rconv{N-D}(h)^\tran & \lconv{D} (\gstar)
\end{pmatrix} \in \RR^{\eta \times \nu},
\]
where the block matrices are linearized accordingly. In our formulation of the
approximate GCRD problem we normalize $\lcoeff_t \lcoeff_\D h$ so that it is a
predetermined constant, which results in essentially the same Jacobian as
described above.

The only difference in the Jacobians, is that the $\nu^{th}$ column would become
the zero column if differentiated with respect to $\lcoeff_t \lcoeff_\D h$,
since $\lcoeff_t \lcoeff_\D h$ is constant. When normalized for computational
purposes, the Jacobian belongs to $\RR^{ \eta \times \nu -1}$ instead (the last
column is deleted).  In the general case when $\gcrd(\fstar,\gstar)=1$, $J$ is
rank deficient by 1 and the $\nu^{th}$ column is a linear combination of the
other columns. The following lemma, similar to \citep[Lemma 4.1]{Zeng11}, formalizes this
statement.

%
\begin{lemma}\label{lem:unique-soln}
  Let $r$ be the residual described earlier with Jacobian $J$. Suppose that
  $\lcoeff_t \lcoeff_\D h$ is a fixed non-zero constant. If
  $\gcrd(\fstar, \gstar) =1$, then all non-zero columns of $J$ are linearly
  independent.
\end{lemma}
\begin{proof}
  Let $ \vec e_\nu \in \RR^{1\times \nu}$ be a unit vector whose last component
  is $1$. We write
  \[
  \vec e_\nu (0,\ldots,0, \pvec h)^\tran = \lcoeff_t \lcoeff_\D h \neq 0.
  \]
  We shall prove the equivalent statement that the matrix
  \[
  \begin{pmatrix}
    J \\
    \vec e_\nu
  \end{pmatrix} =
  \begin{pmatrix}
    \rconv{M-D} (h)^\tran & 	0 		& \lconv{D} (\fstar) \\
    0			  &\rconv{N-D}(h)^\tran	& \lconv{D} (\gstar) \\
    & & \vec e_\nu
  \end{pmatrix} \in \RR ^{(\eta+1) \times \nu}
  \]
  has full rank.

Suppose the converse holds, then there exists $q_1,q_2,p \in
\RR[t][\D;']$ with
$\degD q_1\leq M-D, \degD q_2 \leq N-D$ and $\degD p \leq D$ such
that their
  combined coefficient vector satisfies
  \[
  \begin{pmatrix}
    J \\
    \vec e_\nu
  \end{pmatrix}
  \begin{pmatrix}
    \pvec{q_1}^\tran \\
    \pvec{q_2}^\tran \\
    -\pvec{p}^\tran
  \end{pmatrix}
  =
  \begin{pmatrix}
    0\\
    0\\
    0
  \end{pmatrix}.
  \]
  Expressing this as multiplication over $\RR[t][\D;']$, we have that
  \begin{align*}
    \fstar p &= q_1 h, \\
    \gstar p &= q_2 h.
  \end{align*}
  We conclude that $\gcrd(\fstar p, \gstar p) = p$, as
  $\gcrd(\fstar, \gstar)=1$. If $p=0$ or $q_1=0$ or $q_2=0$, then we are done
(as $\RR[t][\D;']$ is a
  domain). Suppose that $p\neq 0$ and $q_1 \neq 0$ and
$q_2 \neq 0$. Accordingly
we must also have that     $\gcrd(q_1 h, q_2 h) = \gcrd(q_1,q_2) \alpha h = p$
for some $\alpha   \neq 0$. Since $\degD p \leq   \degD h $ it follows that
$\gcrd(q_1,q_2)=1$ so $p=\alpha h$.

%
%
  Since $ p= \alpha h$ we must have that
  $\alpha \fstar = q_1$ and $\alpha \gstar = q_2$.  Now,
  \[
  \vec e_\nu (0,\ldots, 0, \pvec h)^\tran = \lcoeff_t \lcoeff_\D h \neq 0.
  \]
  On the other hand,
  \[
  \vec e_\nu (0,\ldots, 0, \alpha \pvec h)^\tran = 0.
  \]
  This occurs if and only if $\alpha =0$.  But in this case $p=0$ as well, so
  \[
  \begin{pmatrix}
    \pvec{q_1}^\tran \\
    \pvec{q_2}^\tran \\
    -\pvec{p}^\tran
  \end{pmatrix} =
  \begin{pmatrix}
    0 \\
    0 \\
    0 \\
  \end{pmatrix}.
  \]
  It follows that the only vector in the null space is the zero vector, hence
  $ \begin{pmatrix}
    J \\
    \vec e_\nu
  \end{pmatrix}$
  has full rank.  Since any subset of linearly independent vectors is also
  linearly independent, we have that when $\lcoeff_t \lcoeff_\D h$ is a fixed
  non-zero constant that $J$ has rank $\nu-1$.
  \qed
\end{proof}
Note that from the proof we see that if $\lcoeff_t \lcoeff_\D h$ were not fixed,
then the vector $( \pvec \fstar, \pvec \gstar, \pvec h)^\tran$ forms a basis for
the nullspace of $J$. Intuitively, if we did not fix $\lcoeff_t \lcoeff_\D$ in
advance, then there would be infinitely many tuples of $(h,\fstar,\gstar)$ with
the same degree structure over $\RR[t]$ that minimized $\Phi$, since for any
$\alpha \neq 0$ we have
\[
\norm{f-\fstar h}_2^2 + \norm{g-\gstar h}_2^2 = \norm{f- (\alpha \fstar) (
  \alpha ^{-1} h)}_2^2 + \norm{g-(\alpha \gstar) (\alpha ^{-1} h)}_2^2 .
\]
In other words, we need to normalize $h$ in advance to obtain a unique solution.

\begin{corollary}\label{cor:locally-unique-soln}
  Let $r$ be the residual defined earlier in this section with
  $\lcoeff_t \lcoeff_\D h$ a non-zero constant.  If $r= 0$, then the Hessian
  matrix $\nabla^2 \Phi (h,\fstar,\gstar) $ is positive definite.
\end{corollary}
\begin{proof}
  Let $J$ be the Jacobian of $r$. $J$ has full rank, so $J^\tran J$ has full
  rank and is positive semidefinite. If $r=0$, at the global minimum we have
  that $2J^\tran J = \nabla ^2 \Phi$, and $\nabla^2 \Phi (h,\fstar,\gstar)$ is
  positive definite.
  \qed
\end{proof}

When there is no residual, the Hessian $\nabla^2 \Phi (h,\fstar,\gstar)$ is
positive definite. It follows that if $f$ and $g$ are perturbed by a
sufficiently small amount, then $\nabla^2 \Phi$ remains locally positive
definite, and Newton iteration will converge to the (local) global minimum with
an initial guess that is sufficiently close.

We are able to obtain a condition number for a structured perturbation through
the Jacobian of the residuals. Since $J$ has full rank, the smallest singular
value $\sigma_{\nu-1}$ of $J(r(h,\fstar, \gstar))$ is strictly positive. If we
consider structured perturbations, then we are able to show that the approximate
GCRD problem is (locally) well-posed.

In the next lemma, we make use of the fact that for any $f \in \RR[t][\D;']$,
we have that $\norm{f}_2 = \norm{ \pvec {f}}_2$.
%
%
\begin{lemma}
  Let $f,g,h,\fstar,\gstar \in \RR[t][\D;']$ be such that
  $\Phi(h,\fstar,\gstar) <\varepsilon$ for some $\varepsilon >0$, with
  $\lcoeff_t \lcoeff_\D h$ a fixed non-zero constant.  Suppose
  $\fhat, \ghat, \hhat, \fstarhat, \gstarhat \in \RR[t][\D;']$ possess the same
  degree structures as $f,g,h,\fstar $ and $\gstar$ and that
  \[
  \widehat \Phi(\hhat,\fstarhat,\gstarhat)= \norm{\fhat -\fstarhat \hhat}_2^2 +
  \norm{\ghat -\gstarhat \hhat}_2^2 <\varepsilon.
  \]
  Then,
  \[
  \tallnorm{ 
    (h-\hhat, \fstar -\fstarhat, \gstar -\gstarhat)
  }_2^2 \leq \frac{1}{\sigma_{\nu-1}^2} \left ( 2\varepsilon + \tallnorm{
      (f-\fhat, g-\ghat)
    }_2^2 \right ) ~+ ~\parbox{2.4cm}{higher order\\ terms.}
  \]
\end{lemma}


\begin{proof}
  Let $J=J(r(h,\fstar,\gstar))$ be the Jacobian of the residuals from earlier in
  this section.  We have that
  \[
  ( \pvec{\fstar h} - \pvec{\fstarhat \hhat} ,\pvec{ \gstar h} - \pvec{\gstarhat
    \hhat})^\tran
  \approx J ( \pvec {h-\hhat}, \pvec{\fstar}-\pvec{\fstarhat},
  \pvec{g}-\pvec{\gstarhat})^\tran .
  \]
  Ignoring high order terms and using the well known fact that for a (left)
  pseudo inverse $J^+$ of $J$, that $\norm{J^+}_2 = \frac{1}{\sigma_{\nu-1}}$
  gives us
  \begin{align*}
    \tallnorm{(\pvec{\fstar h - \fstarhat \hhat, \gstar h- \gstarhat \hhat})^\tran }_2^2 & \approx
                                                                                           \tallnorm{ J  (\pvec{h-\hhat, \fstar-\fstarhat, \gstar-\gstarhat})^\tran}_2^2  \\ & \geq
                                                                                                                                                                               \sigma_{\nu-1}^2 \tallnorm{(\pvec{h-\hhat,\fstar-\fstarhat,\gstar-\gstarhat)} }_2^2.
  \end{align*} 
  A straightforward application of the triangle inequality gives
  \begin{align*}
    \Big\|(h-\hhat, & \fstar-\fstarhat, \gstar-\gstarhat) \Big\|_2^2 \\
    & \leq \frac{1}{\sigma_{\nu-1} ^2} \tallnorm{(\fstar h - \fstarhat \hhat, \gstar h- \gstarhat \hhat)}_2^2 \\
    & \leq \frac{1}{\sigma_{\nu-1} ^2}\left ( \Phi(h,\fstar,\gstar)+\widehat \Phi(\hhat, \fstarhat, \gstarhat) + \tallnorm{(f-\fhat,g-\ghat)}_2^2 \right)  \\
   & \leq \frac{1}{\sigma_{\nu-1}^2} \left (
    2\varepsilon +
   \tallnorm{(f-\fhat,g-\ghat)
   }_2^2\right) +
     \text{\it higher order terms. \quad\qed}
  \end{align*}
\end{proof}

 \begin{corollary}
   Suppose that $h_{opt}, \fstar_{opt}, \gstar_{opt} \in \RR[t][\D;']$ are a
   locally unique global minimum of $\Phi$ in some neighborhood around
   $h,\fstar$ and $\gstar$.  If
   \[
   \Phi(h,\fstar,\gstar) < \varepsilon \text{ and }\Phi(h_{opt}, \fstar_{opt},
   \gstar_{opt}) < \varepsilon\] for $\varepsilon >0$, then
   \[ \tallnorm{ 
     (h-h_{opt} , \fstar -\fstar_{opt},\gstar-\gstar_{opt})
   }_2^2 \leq \frac{2\varepsilon }{\sigma^2_{\nu-1}} + \text{higher order
     terms}.
   \]
 \end{corollary}
 If we compute different approximate GCRD pairs of $f$ and $g$ (using different
 optimization techniques or initial guesses), then we are able to bound the size
 of the perturbations of $\fstar,\gstar$ and $h$ based on how near they
 are. Furthermore, this corollary allows us to certify an upper bound on the
 distance between our computed approximate GCRD tuple and the actual global
 minimum.

%


\section{Implementation of Approximate GCRD}
\label{chap:implementation}
\label{sec:implementation}

%

This section discusses the particulars and implementation of the algorithms.
The algorithms are described in a Maple-like pseudo code, with Matlab style
matrix indexing.  All of the algorithms have been implemented in the Maple
programming language. For convenience, the notation and assumptions introduced
at the start of Section~\ref{sec:opt-gcrd} will hold, unless otherwise stated.
Additionally, we will assume that content from differential polynomials can be
removed numerically, as computed quantities are typically not primitive due to
round-off errors.

The matrices $S=S(f,g) \in \RR[t]^{(M+N) \times (M+N)}$ will be the differential
Sylvester matrix of $f$ and $g$, and
$\Shat = \Shat(f,g) \in \RR^{ (M+N)(\mu+1) \times (M+N)(\mu+d+1)}$ will be the
inflated differential Sylvester matrix of $f$ and $g$, where $\mu=2(M+N)d$.

The presentation and theoretical analysis of the algorithms is presented in a
bottom-up manner, reflecting their dependencies. Asymptotic upper bounds on the
number of floating point operations required are provided. Furthermore, we
discuss whether the output of the algorithm can be certified in some manner,
when applicable.


We demonstrate the robustness of our algorithms in practice. Specific examples
are provided to thoroughly demonstrate the steps of the algorithms. We
investigate interesting families of input.  In particular, we investigate exact
inputs with an exact GCRD, and perturbed differential polynomials with varying
errors and noise introduced. The test cases of differential polynomials of
interest to us have
\begin{itemize}
\item low degree in $t$ and high degree in $\D$ (unbalanced in $\D$),
\item high degree in $t$ and low degree in $\D$ (unbalanced in $t$), and
\item proportional degrees in $t$ and $\D$ (balanced degrees).
\end{itemize}

\subsection{Algorithms for Approximate GCRD}
\label{sec:algs-for-gcrd}

%
%

We adapt techniques from the exact setting to a numerical setting to compute an
exact GCRD numerically.  These algorithms compute the rank of the differential
Sylvester matrix and a least squares solution to a polynomial linear system,
corresponding to the B\'ezout coefficients.  We describe an algorithm for
finding nearby differential polynomials introduced in
\citep{GiesbrechtHaraldson14}, whose (inflated) differential Sylvester matrix is
nearly singular. Using the least squares numeric GCRD algorithm, we can compute
an approximate GCRD candidate from the nearly singular differential Sylvester
matrix. From this candidate, we extract a guess for the co-factors numerically
and proceed with post-refinement Newton iteration.

%

%

\subsubsection{Numerical Computation of a GCRD}
\label{ssec:numeric-gcrd}
%


Before we can compute a GCRD numerically, the rank of the differential Sylvester
matrix needs to be determined.  Our numeric rank algorithm is an adaptation of
the rank algorithm used by \cite{CGTW95}.  There are
\[
(M+N)(\mu+d+1) - (M+N)(\mu+1) = (M+N)d = \mu/2
\]
trivial singular values\footnote{The inflated differential Sylvester matrix has
more columns than rows, however the nullspace of the columns contains the
information pertaining to the GCRD. The trivial singular values are the zero
singular values occuring from there being more columns than rows.}, and $\mu/2
< \mu+d+1$, the column block size. These
trivial singular values need to be accounted for when annihilating small
singular values.  In the full rank case, we should not underestimate the rank of
$S$ by inferring from $\Shat$, as there are strictly fewer trivial singular
values than the column block size.

\begin{algorithm}[!h]
  \caption{\textbf{:} \texttt{DeflatedRank}}
  \label{alg:RankAlg}

 \begin{algorithmic}[1]
   \smallskip \Require
 \item[$\bullet$] An inflated differential Sylvester matrix
 \[
   \Shat \in \RR^{(M+N)(\mu+1)
     \times(M+N)(\mu+d+1)};
  \]
 \item[$\bullet$] A user defined search radius $\epsilon_{rank} >0$ for
comparing singular
   values.

   \smallskip \Ensure
 \item[$\bullet$] The (scaled) numeric rank $\varrho$ of the (non-inflated)
   differential Sylvester matrix $S$.

   \medskip \State Compute the singular values
   $\sigma_1, \sigma_2,\ldots, \sigma_{(M+N)(\mu+d+1)}$ of $\Shat$ in descending
   order.  \State Find the maximum $k$ such that
   $\sigma_k > \epsilon_{rank} \frac{\sqrt{ (M+N)(2\mu+d+2) }}{\mu+d+1}$ and
   $\sigma_{k+1}<\epsilon_{rank}$.  \State if $\sigma_k > \epsilon_{rank}$ for
all $k$ then
   $\Shat$ has full rank.  \State If there is no significant change (there is no
   maximum $k$) between $\sigma_k$ and $\sigma_{k+1}$ for all $k$, as determined
   by step 2 then return
   failure. 
   \State Set $ \varrho =\left \lceil \frac{k}{\mu+d+1} \right \rceil$, the
   scaled rank of $S$.
 \end{algorithmic}
\end{algorithm}

Algorithm~\ref{alg:RankAlg} computes a reasonable guess for the degree in $\D$
of an approximate GCRD, although it is not generally certifiable.
When $\gcrd(f,g)$ is non-trivial (no errors present in the input coefficients),
we compute (generically) the degree of the GCRD of $f$ and $g$. In the exact
setting, we can
now formulate a linear algebra problem over $\RR[t]$ to compute a GCRD. We
present two solutions to this problem. Algorithm~\ref{alg:NumericGCRD} solves
this problem using linear algebra over $\RR(t)$.
Algorithm~\ref{alg:heuristic-quadratic-gcrd} linearizes the problem over $\RR$
and computes a least squares solution.

\begin{algorithm}[!h]
  \caption{\textbf{:} \texttt{NumericGCRD}}
  \label{alg:NumericGCRD}

 \begin{algorithmic}[1]
   \smallskip \Require
 \item[$\bullet$] $f,g \in \RR[t][\D;']$ with $\norm{f}=\norm{g}=1$;
 \item[$\bullet$] A search radius $\epsilon_{rank}>0$.  \smallskip \Ensure
 \item[$\bullet$] $h= \gcrd(f,g) \in \RR[t][\D;']$ with $\degD h\geq 1$,
 \item[$\bullet$] or an indication that $f$ and $g$ are co-prime within search
   radius $\epsilon_{rank}$.  \medskip \State $M\gets \degD f$,
   $N\gets \degD g$, $d\gets \max\{\deg_t f,\deg_t g\}$, $\mu \gets 2(M+N)d$.
   \State $S\gets S(f,g) \in \RR[t]^{(M+N)\times(M+N)}$.  \State Form the
   inflated differential Sylvester matrix\newline
   $\Shat = \Shat(f,g) \in \RR^{ (M+N)(\mu+1) \times (M+N)(\mu+d+1)}$ of $S$.
   \State Compute the numerical rank $\varrho$ of $S$ using
   Algorithm~\ref{alg:RankAlg} on $\Shat$ with search radius $\epsilon_{rank}$.
   \State If $\varrho>0$, then set $\degD h =D= M+N-\varrho$. Otherwise indicate
   that $f$ and $g$ are co-prime with respect to $\varepsilon_{rank}$ and
   return.  \State Solve for $w \in \RR[t][\D;'] ^{ 1\times (M+N)}$ from
   \[
   w S = (*_1, *_2 , \ldots, *_{D+1},
0,\ldots, 0),
   \]
   ensuring that $\norm{\lcoeff_t(*_{D+1})} \gg 0$.  \State Set
   $(h_0,h_1,\ldots, h_D,0,\ldots,0) = wS$.  \State \Return
   $\content(h)^{-1} h $.
 \end{algorithmic}
\end{algorithm}

\begin{algorithm}[!h]
  \caption{\textbf{:} \texttt{NumericGCRDviaLS}}
  \label{alg:heuristic-quadratic-gcrd}
  \begin{algorithmic}[1]
    \smallskip \Require
  \item[$\bullet$] $f,g \in \RR[t][\D;']$ with $\norm{f}=\norm{g}=1$;
  \item[$\bullet$] $\varepsilon_{rank} >0$ used to compute the degree of the
    GCRD.  \Ensure
  \item[$\bullet$] $h \in \RR[t][\D;']$ that is numerically primitive with a
	fixed leading
    coefficient such that $\norm{w S(f,g) - h}_2^2$ is minimized.
    \State
    $M\gets \degD f$, $N\gets \degD g$, $d\gets \max \{ \deg_t f ,\deg_t g\}$,
    $\mu \gets 2(M+N)d$.  \State Compute $D$ using Algorithm~\ref{alg:RankAlg}
    with $\varepsilon_{rank}$.
    \State $\dvec \gets (\underbrace{\mu+d,\ldots, \mu+d}_{D+1},0,\ldots 0)$ (or
    another valid initial guess).
  \item Compute a least squares solution of $h$ from $\norm{w S(f,g) - h}_2$
    with $\dvec(h) = \dvec$ and $\lcoeff_t \lcoeff_\D h =1$.
  \item $\dvec \gets \dvec(\content(h)^{-1} h)$.
  \item Compute a new least squares solution of $h$ from $\norm{w S(f,g) - h}_2$
    with $\dvec(h) = \dvec$ and $\lcoeff_t \lcoeff_\D h =1$.
    \State \Return $h$.
  \end{algorithmic}
\end{algorithm}

In the implementation of Algorithm~\ref{alg:NumericGCRD}, we take special care
to ensure that $\lcoeff_t \lcoeff_\D h$ does not vanish when $h$ is
normalized. If $\lcoeff_t \lcoeff_\D h$ vanishes, then this could be an
indication that the input is ill-conditioned or content removal of $h$
failed. In either case, it is possible that this instance of the approximate
GCRD problem will not have an attainable global minimum in accordance with
Theorem~\ref{thm:global-min}.




\subsubsection{Nearby Differential Polynomials with GCRD Algorithm}

The matrix $\Shat$ is highly structured, as it is composed of block Toeplitz
matrices.  When we consider the matrix $\Shat + \Delta \Shat,$ the nearest 
(unstructured) matrix of prescribed rank deficiency, we have
considerable flexibility in how we recover the coefficients of $\ftil$ and
$\gtil$, nearby differential polynomials with an exact, non-trivial GCRD
as in \eqref{eq:ftilgtil}. However, the matrix $\Shat + \Delta \Shat$ is not
generally an inflated differential Sylvester matrix, but it is probably
reasonably close to one (see \cite{GiesbrechtHaraldson14,Haraldson15}). We
recall that the mapping $\Gamma:\RR[t] \to \RR^{(\mu+1)\times (\mu+d+1)}$
generates the (rectangular) Toeplitz blocks of $\Shat$.  To recover the
coefficients of $\ftil$ and $\gtil$ one must make a suitable definition for
the mapping $\Gamma^{-1}:\RR^{(\mu+1)\times(\mu+d+1)} \to \RR[t]$.
We use $\Gamma^{-1}$ to find $\ftil,\gtil \in \RR[t][\D;']$ such that
$\Shat(\ftil,\gtil) \approx \Shat(f,g) + \Delta \Shat(f,g)$.

\begin{algorithm}[!h]
  \caption{\textbf{:} \texttt{DeflatedPerturbation}}
  \label{alg:deflated-perturbation}
  \begin{algorithmic}[1]
    \smallskip \Require
  \item[$\bullet$] $f,g \in \RR[t][\D;']$ with $\norm{f}=\norm{g}=1$;
  \item[$\bullet$] Perturbed inflated differential Sylvester matrix
    \[
    \Shat + \Delta \Shat \in \RR^{(M+N)(\mu+1)\times (M+N)(\mu+d+1)};
    \]
  \item[$\bullet$] $\Gamma^{-1}:\RR^{(\mu+1)\times (\mu+d+1)}\to \RR[t]$.

    \smallskip \Ensure
  \item[$\bullet$]$\ftil, \gtil \in \RR[t][\D;']$ where
    $\dvec(\ftil) \leq \dvec(f)$ and $\dvec(\gtil) \leq \dvec(g)$.  \medskip
    \State $M\gets \degD f$, $N\gets \degD g$,
    $d\gets \max\{\deg_t f,\deg_t g\}$, $\mu \gets 2(M+N)d$ and
    $N\gets \mu+d+1$.
    \For {$0\leq i \leq \degD f$ } \State
    $[I,J] \gets [1:\mu+1][(i+1)+(i-1)N(\mu+d+1) : (i+1)N(\mu+d+1)]$ \State
    $\ftil_i \gets ~\Gamma^{-1} \left( (\Shat +\Delta \Shat)[I,J] \right)$
    \EndFor
    \For {$0\leq i \leq \degD g$ } \State
    $[I,J] \gets [N(\mu+1)+1 : (N+1)(\mu+1)][(i+1)+(i-1)M(\mu+d+1) :
    (i+1)M(\mu+d+1)]$
    \State
    $\gtil_i \gets ~{\Gamma^{-1} \left( (\Shat +\Delta \Shat)[I,J] )\right)}$
    \EndFor
    \State \Return $\ftil$ and $ \gtil$.
  \end{algorithmic}
\end{algorithm}

Regardless of our choice of $\Gamma^{-1}$, this method of recovering $\ftil$ and
$\gtil$ can lead to a differential Sylvester matrix that does not have the
desired numeric rank, as determined by Algorithm~\ref{alg:RankAlg}. The
perturbation $\Delta \Shat$ is unstructured while $\Gamma (\ftil_i)$ and
$\Gamma (\gtil_j)$ are (highly structured) Toeplitz matrices.  Consequently,
some non-zero terms of $\Delta \Shat$ are ignored.

\begin{algorithm}[!h]
  \caption{\textbf{:} \texttt{NearbyWithGCRD}}
  \label{alg:SVD-GCRD}

 \begin{algorithmic}[1]
   \smallskip \Require
 \item[$\bullet$] $f,g \in \RR[t][\D;']$ with $\norm{f}=\norm{g}=1$;
 \item[$\bullet$] A search radius $\epsilon_{rank}>0$, used to validate the
   degree of $h$.  \Ensure
 \item [$\bullet$] $\ftil,\gtil \in \RR[t][\D;']$ where
   $\dvec( \ftil) \leq \dvec (f)$, $(\dvec(\gtil) \leq \dvec(g)$ and
   $h \approx \gcrd(\ftil,\ftil) \in \RR[t][\D;']$ with $\degD h\geq 1$, or ;
 \item [$\bullet$] An indication that $f$ and $g$ are co-prime within search
   radius $\epsilon_{rank}$.  \medskip \State $M\gets \degD f$,
   $N\gets \degD g$, $d\gets \max\{\deg_t f,\deg_t g\}$ and $\mu \gets 2(M+N)d$.
   \State $S\gets S(f,g) \in \RR[t]^{(M+N)\times(M+N)}$.  \State
   $\Shat \gets \Shat(f,g) \in \RR^{ (M+N)(\mu+1) \times (M+N)(\mu+d+1)}$.
   \State Compute the SVD of $\Shat$, where $\Shat = P\Sigma Q$.
   \State Compute the numerical rank $\varrho$ of $S$ using
   Algorithm~\ref{alg:RankAlg} on $\Shat$ with search radius $\epsilon_{rank}$.
   \State If $\varrho>0$ set the last $\varrho(\mu+d+1)$ singular values to $0$
   and compute $\Sigmabar$.
   Otherwise indicate that $f$ and $g$ are co-prime with respect to
   $\epsilon_{rank}$.  \State Compute $\Shat + \Delta \Shat = P \Sigmabar Q$.
   \State Compute $\ftil$ and $\gtil$ from $\Shat + \Delta \Shat$ using
   Algorithm~\ref{alg:deflated-perturbation}.  \State Compute
   $h =\mathtt{NumericGCRD}(\ftil,\gtil)$ using
   Algorithm~\ref{alg:heuristic-quadratic-gcrd}, with $\epsilon_{rank}$ used to
   validate the degree of $h$ using Algorithm~\ref{alg:RankAlg}.  \State \Return
   $\ftil, \gtil$ and $h$.
 \end{algorithmic}
\end{algorithm}


\subsubsection{Numeric Right Division}

Numeric right division without remainder between two differential polynomials is
a rational function linear algebra problem. The (approximate) quotient is a
solution to a linear system, in a least squares sense.
We present a naive algorithm that works well in practice and a more rigorous linear least squares variant.

The solution
to this system may not be in (approximate) lowest terms. In our implementation we use approximate GCD
and real linear least squares to resolve this.  We note that total least squares
can also be employed to prevent the need of an approximate
GCD computation to put the rational function coefficients in lowest terms.

\begin{algorithm}[!h]
  \caption{\textbf{:} \texttt{NaiveNumericRightDivision}}
  \label{alg:numeric-right-division}

 \begin{algorithmic}[1]
   \smallskip \Require
 \item[$\bullet$] $f,h \in \RR[t][\D;']$ with $\norm{f}=\norm{h}=1$.
   \Ensure \item [$\bullet$] $\fstar \in \RR(t)[\D;']$ satisfying $f=\fstar h$.
   \State $M\gets \degD f$, $D\gets \degD h$.
   \State Form the matrix $\calM(h)$
   from the last $M-D+1$ columns of $\rconv{M-D}(h)$.
\State Solve
   \[
   (f_D,f_{D+1},\ldots,f_M) = (\fstar_0,\fstar_1,\ldots,\fstar_{M-D}) \calM(h)
   \]
   by backwards substitution for the coefficients of $\fstar$.
   \For
   {$0\leq i \leq M-D$} \State $\fstar_i \gets \text{ Approximate $\fstar_i$ in
     rational function least terms}$
   \EndFor
   \State \Return $\fstar$.
 \end{algorithmic}
\end{algorithm}


\begin{algorithm}[!h]
  \caption{\texttt{:} \texttt{NumericRightDivisionViaLS}}
  \label{alg:numeric-right-division-via-ls}
  \begin{algorithmic}[1]
    \smallskip \Require
  \item[$\bullet$] $f,h \in \RR[t][\D;']$ with $\norm{f}=\norm{h}=1$.
    \Ensure
  \item[$\bullet$] $\fstar \in \RR(t)[\D;']$ in lowest terms satisfying
    $f=\fstar h$.
    \State $M\gets \degD f$, $D\gets \degD h$.
    \State Solve
    \[
    v_{-1}(f_0,f_{1},\ldots,f_M) = (v_0,v_1,\ldots,v_{M-D})\rconv{M-D}(h)
    \]
    by linear least squares for the coefficients of $v_{-1},v_0,\ldots, v_{M-D}$.
    \For
   {$0\leq i \leq M-D$}
   \State $\fstar_i \gets \text{ Approximate $\dfrac{v_i}{v_{-1}}$ in
     rational function least terms}$
   \EndFor
   \State \Return $\fstar$.
  \end{algorithmic}
\end{algorithm}


\subsubsection{Improved GCRD via Optimization: Newton's Method}

Using Algorithm~\ref{alg:SVD-GCRD}, we can compute an initial guess for an
approximate GCRD, $h_{init}$. We can perform right division without remainder
numerically to compute initial guesses for the co-factors, $\fstar_{init}$ and
$\gstar_{init}$.  We now have enough information to set up a post-refinement
Newton iteration, to hopefully compute an approximate GCRD.  When the co-factors
have polynomial coefficients, the products $\fstar h$ and $\gstar h$ are always
polynomial. This makes Newton iteration a very straightforward procedure, as the
objective function
\[
\Phi(h,\fstar,\gstar) = \norm{f-\fstar h}_2^2 + \norm{g-\gstar h}_2^2
\]
is easily computed.  However, when the co-factors have rational function
coefficients, the quantities $\fstar h$ and $\gstar h$ usually have rational
function coefficients due to round-off error.
We can clear fractions and compute the
least squares solution of an equivalent associate problem.

\begin{algorithm}[!h]
  \caption{\textbf{:} \texttt{NewtonIteration}}
  \label{alg:Newton}

 \begin{algorithmic}[1]
   \smallskip \Require
 \item[$\bullet$] $f,g,h_{init} \in \RR[t][\D;']$ with
   $\norm{f}=\norm{g}=\norm{h_{init}}=1$;
 \item[$\bullet$] $k \in \NN$, the number of iterations.  \Ensure
 \item[$\bullet$] $\fstar, \gstar \in \RR[t][\D;']$ and $h \in \RR[t][\D;']$
   such that $\Phi( \fstar h, \gstar h)$ is locally minimized and
 \item [$\bullet$] $\dvec(\fstar h) \leq \dvec (f)$ and
   $\dvec(\gstar h) \leq \dvec (g)$.  \State $M\gets \degD f$, $N\gets \degD g$
   and $D\gets \degD h_{init}$.  \State Compute initial guesses of $\fstar$ and
   $\gstar$ using Algorithm \ref{alg:numeric-right-division-via-ls}.  \State
   $\lcoeff_t \lcoeff_\D h \gets \lcoeff_t \lcoeff_\D
   h_{init}$. 
   \State
   $x^0 \gets (\pvec{\fstar_{init}}, \pvec{\gstar_{init}},\pvec{h_{init}})
   ^\tran$.
   \For { $1\leq i\leq k$} \State Solve
   $\nabla ^2 \Phi(x^i) \cdot x^{i+1} = \nabla^2 \Phi(x^i)\cdot x^i - \nabla
   \Phi(x^i)$ for $x^{i+1}$.
   \EndFor
   \State \Return $\fstar$, $\gstar$, and $h$ computed from $x^k$.
 \end{algorithmic}
\end{algorithm}

The normalization we impose, that $\lcoeff_t \lcoeff_\D h$ is fixed, ensures the
solution is (locally) unique, by Corollary~\ref{cor:locally-unique-soln}.  We
note that this normalization can be changed. However one must ensure that the
normalization vector is not orthogonal to
$(\pvec \fstar, \pvec \gstar, \pvec h)$.  We now generalize the Newton iteration
for the instance when the co-factors have rational function coefficients.

\begin{algorithm}[!h]
  \caption{\textbf{:} \texttt{ModifiedNewtonIteration}}
  \label{alg:Modified-Newton}

 \begin{algorithmic}[1]
   \smallskip \Require
 \item[$\bullet$] $f,g,h_{init} \in \RR[t][\D;']$ with
   $\norm{f}=\norm{g}=\norm{h_{init}}=1$;
 \item[$\bullet$] $k \in \NN$, the number of iterations.  \Ensure
 \item[$\bullet$] $\fstar, \gstar \in \RR(t)[\D;']$ and $h \in \RR[t][\D;']$
   such that $\Phi( \fstar h ,\gstar h)$ is locally minimized and
 \item [$\bullet$] $\dvec(\fstar h) \leq \dvec (f)$ and
   $\dvec(\gstar h) \leq \dvec (g)$.  \State $M\gets \degD f$, $N\gets \degD g$
   and $D\gets \degD h_{init}$.  \State Compute initial guesses of $\fstar$ and
   $\gstar$ using Algorithm \ref{alg:numeric-right-division-via-ls}.  \State
   $\lcoeff_t \lcoeff_\D h \gets \lcoeff_t \lcoeff_\D h_{init}$.  \State
   $f\gets \fstar_{-1} f, g\gets \gstar_{-1} g, \fstar \gets \fstar_{-1} \fstar$
   and $\gstar \gets \gstar_{-1} \gstar$.
   \State
   $x^0\gets (\pvec{\fstar_{init}}, \pvec{\gstar_{init}},
   \pvec{h_{init}})^\tran$.
   \For { $1\leq i\leq k$} \State Solve
   $\nabla ^2 \Phi(x^{i+1}) \cdot x^i = \nabla^2 \Phi(x^i)\cdot x^i - \nabla
   \Phi(x^i)$ for $x^{i+1}$.
   \EndFor
   \State $\fstar \gets \frac{1}{\fstar_{-1}} \fstar $ and
   $\gstar \gets \frac{1}{\gstar_{-1}} \gstar$.  \State \Return
   $ \fstar, \gstar $ and $ h$ computed from $x^k$.
 \end{algorithmic}
\end{algorithm}

\subsection{Analysis of Algorithms}
In this section we assess the computational cost in terms of the number of
floating point operations or {\emph {flops}}. Where applicable, we discuss the
numerical stability of the algorithms and whether or not their output can be
certified. The algorithms are analyzed in the order they were presented. The
assumption that content can be removed numerically is not without loss of
generality; content removal can be unstable if implemented poorly.

In our implementation we remove content by (re)formulating our solutions as a
solution to a (total) least squares problem.  This can be done by performing the
SVD on a generalized Sylvester matrix of several univariate polynomials
\citep{KYZ06} to infer the degree of the content.  Computing the degree of the
content this way generalizes the method of \cite{CGTW95} to several
polynomials. In our implementation the only important information is the degree
of an approximate GCD, so we assume that the run-time of approximate GCD is
cubic in the number of variables. One could compute an approximate GCD of
several polynomials and perform a least squares division, however
post-refinement would likely be needed. We generally assume that unstructured
linear algebra techniques are used on the problems, however structured methods
could lead to a modest asymptotic improvement.

\subsubsection{Analysis of Algorithm \ref{alg:RankAlg} -- {\tt DeflatedRank}}

The number of flops Algorithm~\ref{alg:RankAlg} requires is dominated by the
cost of performing the SVD on $\Shat$. The SVD requires
$O( (M+N)^3(\mu+d+1)^3) = O( (M+N)^6 d^3)$ flops, using standard arithmetic.  As
mentioned earlier, this algorithm is generally not certified to produce the
degree of an approximate GCRD.

\subsubsection{Analysis of Algorithm \ref{alg:NumericGCRD} -- {\tt NumericGCRD}}

The number of flops Algorithm~\ref{alg:NumericGCRD} requires is ultimately
bounded by the cost of computing the rank of $S$ using
Algorithm~\ref{alg:RankAlg}. The cost of Algorithm~\ref{alg:RankAlg} is
$O( (M+N)^6d^3)$ flops. The cost of computing a GCRD given the degree in $\D$ is
$O((M+N)^3)$ operations over $\RR(t)$ which corresponds to $O( (M+N)^3d^2)$
flops.  The cost of the approximate GCD and division to remove content depends
on the specific method used, but is usually negligible when compared to the rank
computation.

This algorithm is not numerically stable for large degree inputs in $t$ and
$\D$.  Performing linear algebra over $\RR(t)$ leads to considerable degree
growth in $t$, and removing (approximate) content with a division further
perturbs the coefficients of the GCRD. The output of this algorithm is not
certified to be correct in most instances.

\subsubsection{Analysis of Algorithm \ref{alg:heuristic-quadratic-gcrd} -- {\tt NumericGCRDviaLS}}
There are $(M+N)(\mu+d+1)$ equations and
$(M+N)(\mu+1) + (D+1)(\mu+d+1) = O( (M+N)^2d )$ unknowns.  The cost of computing
the least squares solution is $O( (M+N)^6 d^3)$ flops.  The cost of inferring
the content by looking at singular values of the (generalized) Sylvester matrix
is bounded by $O( (D)^3(\mu+d+1)^3) = O((M+N)^6d^3)$ flops. The total number of
flops required for this algorithm is $O((M+N)^6 d^3)$.

This algorithm relies on solving a real linear least squares problem. As such,
this algorithm is numerically stable, provided that the underlying least squares
problem is reasonably conditioned, and solved in a reasonable way. One such
method of solving the least squares problem is the SVD and arising
pseudo-inverse. We are able to certify the correctness of the answer obtained
via least squares, provided that the underlying approximate GCD algorithm
computes the degree of the content correctly.

\subsubsection{Analysis of Algorithm \ref{alg:deflated-perturbation} -- {\tt DeflatedPerturbation}}
The number of flops Algorithm~\ref{alg:deflated-perturbation} requires is
$O( (M+N)^2 d^2 )$, assuming that $\Gamma^{-1}$ uses the weighted block
average. We use this in our implementation.
This algorithm is not certified to provide meaningful output.

\subsubsection{Analysis of Algorithm \ref{alg:SVD-GCRD} -- {\tt NearbyWithGCRD}}
The number of flops Algorithm~\ref{alg:SVD-GCRD} requires is dominated by the
cost of computing the singular values of $\Shat$, which is $O( (M+N)^6 d^3)$
flops.

This algorithm is exactly the same as Algorithm~\ref{alg:NumericGCRD} when
$\Shat$ has the desired rank deficiency. In the event that the input is
approximate, the quality of our answer depends on the largest singular value of
$\Shat$ that we annihilate.
This algorithm is not certified to provide meaningful output, but if used in
conjunction with Algorithm~\ref{alg:heuristic-quadratic-gcrd}, the output can be
certified as a least squares approximation to the solution of the B\'ezout
coefficients.

\subsubsection{Analysis of Algorithm \ref{alg:numeric-right-division} -- {\tt NaiveNumericRightDivision}}
The number of flops Algorithm~\ref{alg:numeric-right-division} requires depends
on the method used to solve the linear system.  The particular system is highly
structured so we can solve it by backwards substitution directly, which costs
$O( (M-D)^2)$ operations over $\RR(t)$. This corresponds to $O( (M-D)^2 d^2 )$
flops. An upper bound on the degree required for approximate GCD computations is
$ (M-D+1)d$. The total cost of each approximate GCD computation is at most
$O(((M-D)d)^3)$ flops. There are at most $M-D+1$ approximate GCD
computations performed, so the total cost of the algorithm is
$O( (M-D)^4d^3)$
flops.

The output of this answer is generally only certifiable if the residual of a least squares division is zero, i.e. the coefficients are exact.
If we
assume that $\lcoeff_t \lcoeff_\D h=1$ and $\norm{h}$ is not arbitrarily large,
then the backwards substitution is well conditioned.  The approximate GCD
computations and following divisions can perturb the coefficients, so the
algorithm can be unstable for poorly conditioned inputs. This is especially
problematic when $\lcoeff_\D h$ is poorly conditioned.

%
%
%

\subsubsection{Analysis of Algorithm \ref{alg:numeric-right-division-via-ls} --
  {\tt NumericRightDivisionViaLS}}

If $\fstar$ has polynomial coefficients, then
$\deg_t \fstar \leq \deg_t f\leq d$ as $f$ and $h$ have polynomial coefficients
as well. If $\fstar$ has rational function coefficients, we recall from Section~\ref{ssec:diff-division} that there are
$O(M(M-D)d)$ equations and $O( (M-D)^2d)$ unknowns.
%
The cost of
solving this linear least squares problem is
$O( (M(M-D)d)^3) \subseteq O(M^6d^3)$ flops.

The output of this algorithm is certified as a linear least squares
solution. Like Algorithm~\ref{alg:numeric-right-division}, the conditioning of
this algorithm is strongly related to the conditioning of $h_D$.

\subsubsection{Analysis of Algorithms
  \ref{alg:Newton}--\ref{alg:Modified-Newton} -- {\tt NewtonIteration} \\ and {\tt
    ModifiedNewtonIteration}}

We transform the problem of computing a GCRD to that of optimizing
$\Phi:\RR[t][\D;']\times \RR(t)[\D;']^2 \to \RR$. We can assume without loss of generality that
$\fstar$ and $\gstar$ have polynomial coefficients, as we can solve an equivalent associate problem instead.
The dominating cost of the Newton iteration is solving a linear system to get
the next value which requires $O(\nu ^3)$ operations, where $\nu$ is the number
of variables needed to represent the coefficients of $h,\fstar$ and $\gstar$.

Newton iteration can fail for many reasons, (it is, afterall, a locally
convergent method) however our Newton iteration usually
fails because:
\begin{itemize}
\item [1.] $\nabla^2 \Phi$ is positive semidefinite at a point in the iteration,
  the stationary point is a saddle point;
\item [2.] The initial guess is poorly chosen and $\nabla^2 \Phi$ is indefinite
  at a point.
\end{itemize}
In the event that Newton iteration fails we can perform a Gauss-Newton iteration
instead.  Despite Gauss-Newton iteration having at least linear convergence, $J^TJ$ is
positive definite, so saddle points are no longer a problem if the optimal
residual is sufficiently small.
According to Corollary~\ref{cor:locally-unique-soln}, if the residual is
sufficiently small then Newton iteration will converge to a global minimum.
\subsection{Examples and Experimental Results} \label{sec:experimental-data}

This section contains some examples of our
implementation.%
\footnote{A proof-of-concept implementation of the algorithms is available
at \url{https://www.scg.uwaterloo.ca/software/ApproxOreFoCM-2019.tgz}.}  The (inflated)
differential Sylvester matrix is ill-conditioned for large degree inputs in $t$
and $\D$. This ill-conditioning occurs because the columns (rows) become
unbalanced due to the falling factorials, where some columns have a Frobenius
norm factorially larger than others.  We restrict ourselves to modest
examples with minimal coefficient
growth. Computations are done using the default precision in \texttt{Maple},
which is approximately 10 decimal points of accuracy.


\begin{example}[No Noise, many factors]
  \begin{align*}
    f= &   .00769\D^5+(.00035 t^2+.05386 t-.05386)\D^4 \\+
       &(.00140 t^3+.06820 t^2-.16928 t+.17313)\D^3\\+
       &(-.09513 t^3+.22559 t^2+.16928 t-.33472)\D^2\\+
       &(.18607 t^3-.65720 t^2-.04617 t+.32702)\D\\ +
       &(-.09234 t^3+.36305 t^2 -.00769 t-.11927). \\
    g= & (.01001 t-.01001)\D^5+(.04019 t^2-.07007 t+.03003)\D^4\\+
       &(.00063 t^3-.01048 t^2+.15014 t-.11010)\D^3\\+
       &(.27901 t^3-.32921 t^2-.09008 t+.17016)\D^2\\+
       &(-.55990 t^3+.52909 t^2-.04004 t-.08007)\D  \\+
       &(.28026 t^3-.22959 t^2+.04004 t).
  \end{align*}

  \noindent
  We compute initial guesses (removing content numerically where appropriate):
  \begin{align*}
    h_{guess} 	=&  .09285\D^3+(.37139t-.27854)\D^2+(-.74278t+.27854)\D+(.37139t-.09285) ,\\
    \fstar_{guess}=&  .08287\D^2+(.00377t^2+.24862t-.33150)\D +\\
                 & (-2.05844\times10^{-10}t^3-.24862t^2+.91162t-.04144), \\
    \gstar_{guess}=&  (.10780t-.10780)\D^2+(.00168t^2+8.67540 \times 10^{-9}t-2.71283\times10^{-9})\D \\ +
                 & (2.35115\times 10^{-8}t^3+.75463t^2-.43122t+6.78976\times10^{-8}).
  \end{align*}
  The quality of this initial guess is
  \[
  \norm{f-\fstar_{guess}h_{guess}}_2^2 +\norm{g-\gstar_{guess}h_{guess}}_2^2 =
  4.04506 \times 10^{-14}.
  \]
  The condition number for the Hessian matrix evaluated at our initial guess is
  $18354.38336$ and our smallest eigenvalue is $.00314$. Since $\nabla^2 \Phi$
  is locally positive definite, we know that we will converge to a unique
  (local) minimum.  The minimum we converge to is $2.33030\times 10^{-20}$.

  The exact GCRD in this example is $h= (\D+4t-1)(\D-1)(\D-1)$.
\end{example}

\begin{example}[Noise]
  In this example we introduced normalized noise of size $10^{-5}$ to $f$ and
  $g$.
  \begin{align*}
    f  = &
           .00583\D^5\\
    +&(-9.45614\times 10 ^{-7}t^3+.00027t^2+.03498t-.03498)\D^4\\
    +&(-8.26797\times 10 ^{-7}t^5+.04743t^3+.01113t^2-.05247t+.07287)\D^3\\
    +&(-9.08565\times 10 ^{-8}t^5+.13885t^4-.21623t^3+.30950t^2-.17781t-.05247)\D^2\\
    +&(-.18655t^5-.02226t^4-.20166t^3-.41974t^2+.33812t-.10202)\D\\
    +&(.18655t^5-.30315t^4+.43935t^3-.22868t^2-.13117t+.15740).\\
    g  = &
           (.00780t-.00779)\D^5\\
    +&(9.10928\times 10 ^{-7}t^5+6.83196\times 10 ^{-7}t^3+.02351t^2-.07018t+.02729)\D^4\\
    +&(5.94796\times 10 ^{-8}t^4+.02376t^3-.07822t^2+.12086t-.06238)\D^3\\
    +&.16326t^4+.04654t^3-.27267t^2+.12476t+.03898)\D^2\\
    +&(-.21833t^5-.10868t^4-.05617t^3+.63939t^2-.38597t+.14036)\D\\
    +&(.21833t^5-.27291t^4-.01462t^3-.09418t^2+.24952t-.12086).
  \end{align*}

  \noindent
  We compute initial guesses (removing content numerically where appropriate):

  \begin{align*}
    h_{guess} = 		 & .11192\D^3 +(.33514t-.22357)\D^2 \\+
                                 & (-.44667t^2-.22358t-.11327)\D+.44754t^2-.55869t+.22453 ,\\
    \fstar_{guess} =	 & (5.97992\times10^{-8}t^5-.00001t^4+.05212-8.67362\times10^{-18}t^2+5.20417\times10^{-18}t)\D^2\\ +
                                 & (-.0001t^5-.00002t^4-.00001t^3+.00238t^2+.15646t-.20842)\D \\ +
                                 & (.00003t^5- +.41663t^3-.15629t^2+.57193t-.02463),\\
    \gstar_{guess} =	 & (-2.10091\times10^{-8}t^5+-6.93889\times 10^{-18}t^3-1.73472\times10^{-17}t^2+.06967t-.06963)\D^2\\
    +&(.00002t^4+.00001t^3+.00146t^2-.27937t+.10474)\D\\+
                                 & (-.00004t^5+.00002t^4+.48596t^3+.00189t^2-.27763t-.00158).
  \end{align*}
  The quality of this initial guess is
  \[
  \norm{f-\fstar_{guess}h_{guess}}_2^2 +\norm{g-\gstar_{guess}h_{guess}}_2^2 =
  .00003.
  \]
  The condition number for the Hessian matrix evaluated at our initial guess is
  $21971.20356$ and our smallest eigenvalue is $.00818$. Since $\nabla^2 \Phi$
  is locally positive definite, we know that we will converge to a unique
  (local) minimum.  The minimum we converge to is $1.06759\times10^{-10}$, which
  is roughly the amount of noise we added.
\end{example}

\begin{example}[GCRD via LS]
  In this example we added a noise factor of $10^{-4}$ to $f$ and
  $g$. Performing Linear Algebra over $\RR(t)$ produced completely unacceptable
  answers, so we used a Least Squares algorithm to compute an approximate GCRD.
  \begin{align*}
    f = & (.11329t^6+.23414t^5+.12840t^4+.00755t^3+.00005)\D^3 \\
    +&(.00001t^6+.23414t^5+.59667t^4+.02269t^3-.04528t^2-.02266t+3.67436\times 10^{-7})\D^2 \\
    +&(-.11329t^6+.33231t^5-.43054t^4-.00754t^3-.00003t^2-.06798t+.00003)\D \\
        &(-.00001t^6-.23414t^5+.34741t^4+.01510t^3-.06799t^2+.09064t+.00004). \\
    g = &(.01938t^4-.03876t^3-.07752t^2+.03876t+.05819)\D^3 \\
    +&(.13567t^4+.23252t^3-.07750t^2-.34879t+.29066)\D^2\\
    +&(-.01938t^4+.13563t^3+.03873t^2+.25195t-.23257)\D\\
    +&(-.13562t^4+.44570t^3-.56198t^2-.03874t+.17439).
  \end{align*}

  \noindent
  Using a Least Squares variant of our Numeric GCRD algorithm, we are able to
  compute (without removing content):
  \begin{align*}
    h_{guess} =  	   & (t^2+1.94162t+.93768)\D^2 + 2.87182\D \\ +&(-.94502t^2+2.84696t-3.82712), \\
    \fstar_{guess} = &(.00712t^6-.01655t^5+.71917t^4+.05630t^3-.00048t^2-.00199t+.00155)\D \\+
                           & (-.00061t^6-.01248t^5+.02319t^4+.04309t^3-.02430t^2-.12387t-.00820), \\
    \gstar_{guess } =& (.00041t^4-.00715t^3+.13885t^2-.50330t+.36942)\D\\+
                           & (.00381t^4-.01139t^3+.86465t^2-.48856t+.00231) .
  \end{align*}
  The quality of this initial guess is
  \[
  \norm{f-\fstar_{guess}h_{guess}}_2^2 +\norm{g-\gstar_{guess}h_{guess}}_2^2 =
  .00328.
  \]
  The condition number for the Hessian matrix evaluated at our initial guess is
  $148.62547$ and our smallest eigenvalue is $.04615$. Since $\nabla^2 \Phi$ is
  locally positive definite, we know that we will converge to a unique (local)
  minimum.  The minimum we converge to is $9.53931 \times 10^{-9}$.
\end{example}

\subsection{General Examples}
We provide results that demonstrate the robustness of our algorithms.  We
consider differential polynomials whose degrees in $t$ and $\D$ are balanced and
unbalanced. The coefficients of the inputs were generated using the Maple
routine \texttt{randpoly()}. The inputs $f$ and $g$ were normalized so that
$\norm{f}=\norm{g}=1$. We introduced normalized noise to the coefficients of $f$
and $g$, so that the relative error is size of the perturbation. Precisely, if
$f+\Delta f$ and $g+\Delta g$ are perturbed from $f$ and $g$ by the quantities
$\Delta f$ and $\Delta g$, then the relative error in the coefficients of $f$
and $g$ is given by $\norm{\Delta f}_2= \norm{\Delta g}_2$.

We recall that the Newton iteration optimizes
$\norm{f-\ftil}_2^2 + \norm{g-\gtil}_2^2$, which is the {\emph{sum of the
    squares}} of the errors. The initial error and error from post-refinement
are expressed as the sum of square errors accordingly. In our
experiments, the \emph{Initial Error} is the
quantity $\norm{f-f_{init}}_2^2 +
\norm{g-g_{init}}_2^2$ and the error after post-refinement, \emph{Newton
Error} is the
quantity $\norm{f-f_{opt}}_2^2 +
\norm{g-g_{opt}}_2^2$. In all of the examples
when there were no perturbations in the coefficients of $f$ and $g$, our numeric
GCRD algorithm and post-refinement procedures were able to compute an exact GCRD
to machine precision.  When perturbations imposing a relative error of $10^{-8}$
in the coefficients of $f$ and $g$ were introduced, we were able to compute a
solution to the approximate GCRD problem in every example.

Introducing perturbations imposing a relative error of order $10^{-4}$ and
$10^{-2}$ into the coefficients of $f$ and $g$ prevented computation of an
approximate GCRD in some examples. Instead, we provide examples of the largest
perturbation in the coefficients of $f$ and $g$ that we were able to compute an
approximate GCRD. Instances that are denoted as ``FAIL'' occur when the
post-refinement did not converge. The implementation of Newton's method is not
globalized to converge to a stationary point, hence the iterates may diverge.
In our examples iterates diverge because the Hessian matrix is indefinite at an
initial guess.

\subsubsection{Balanced Degrees in $t$ and $\D$}
The following results of experiments were conducted on differential polynomials
whose degrees in $t$ and $\D$ were proportional, or balanced.

\medskip\noindent
\begin{tabular}{cccccc}
  Example& Input  $(\D, t)$ & GCRD  $(\D,t)$ & Noise    & Initial Error  & Newton Error \\ \hline
  1		& (2,2)        & (1,1)       & 1e-2 &  2.63579e-3               & 9.37365e-5  \\
  2		& (2,2)        & (1,1)       & 1e-2 &  6.98136e-4               & 8.96068e-5\\
  3		& (3,2)        & (2,1)       & 1e-2 &  1.69968e-2               & 1.26257e-4\\
  4		& (3,4)        & (2,2)       & 1e-2 &  3.8269e-3                & 1.04271e-4\\
  5		& (4,4)        & (3,2)       & 1e-2 &  3.15314e-1               & FAIL \\ \hline
  5		& (4,4)        & (3,2)       & 1e-4 &  9.29336e-7               & 8.97294e-9
\end{tabular}

\subsubsection{Unbalanced Degrees in $\D$ }
The following results of experiments were conducted on differential polynomials
whose degrees in $\D$ were relatively larger than their degree in $t$.

\medskip\noindent
\begin{tabular}{cccccc}
  Example & Input  $(\D, t)$ & GCRD  $(\D,t)$ & Noise    & Initial Error  & Newton Error \\ \hline
  1& 	 (2,2)      	   & (1,1)          & 1e-2     & 1.13109e-3    
  &2.90713e-5                 \\
  2& 	 (3,2)       	   & (2,1)          & 1e-2     & 6.72179e-4     &1.13998e-4	            \\
  3& 	 (4,2)       	   & (3,1)          & 1e-2     & 3.00365e-4     &1.04038e-4                \\
  4& 	 (5,2)       	   & (4,1)          & 1e-2     & 9.01982e-4     &1.23557e-4               \\
  5& 	 (6,2)       	   & (5,1)          & 1e-2     & 6.61552e-3	& FAIL		\\ \hline
  5& 	 (6,2) 		   & (5,1)          & 1e-4     & 2.74084e-4     & 
1.12566e-8
\end{tabular}

\subsubsection{Unbalanced Degrees in $t$ }
The following results of experiments were conducted on differential polynomials
whose degrees in $t$ were relatively larger than their degree in $\D$.

\medskip\noindent
\begin{tabular}{cccccc}
  Example & Input  $(\D, t)$ & GCRD  $(\D,t)$ & Noise    & Initial Error  & Newton Error \\ \hline
  1& 	 (2,3)      	   & (1,2)          & 1e-2     & 1.27092e-2     &1.43153e-4                \\
  2& 	 (2,6)       	   & (1,4)          & 1e-2     & 5.04286e-1     &FAIL	            \\ \hline
  2&	 (2,6)		   & (1,4)	    & 1e-4     & 7.78993e-4	&1.31180e-8 \\
  3& 	 (2,8)       	   & (1,6)          & 1e-4     & 6.9361e-2 	&FAIL                \\ \hline
  3& 	 (2,8)       	   & (1,6)          & 1e-8     & 3.92268e-10    &1.15653e-16                \\
  4& 	 (2,11)       	   & (1,8)          & 1e-8     & 6.20749e-10  	&1.26549e-16           \\
  5& 	 (2,13)       	   & (1,10)         & 1e-8     & 2.23588e-10	&1.03136e-16
\end{tabular}


\section{Conclusion}


In this paper we have formally defined an approximate GCRD problem for
differential polynomials, and given an approach to a robust numerical
solution. We have seen that, under reasonable assumptions the approximate GCRD
problem is well posed. In particular, we show that Newton iteration will
converge to an optimal solution if the residual is sufficiently small. We employ
the earlier results in \citep{GiesbrechtHaraldson14}, analogous to SVD-based
approximate GCD methods like \cite{CGTW95}, to compute a reasonable initial
estimate for the Newton iteration. The results were presented for real
differential polynomials, however the results generalize in a very straight
forward way to the instance of complex differential polynomials.

We believe that some aspects of our problems could also be approached from a
structured low-rank approximation viewpoint \cite{KYZ05, schost15}. In
particular, the work of \cite{schost15} can be used to obtain an initial
low-rank differential Sylvester matrix in which co-factors and a GCRD can be
extracted for post-refinement.  This holds more generally than differential
polynomials, and a particular example to consider is the shift operator,
commonly associated with linear difference equations.

Another area of future work is in the certification of the degree of an
approximate GCRD.  We can obtain a reasonable guess by enumerating over the
degrees of all possible approximate GCRDs, similar to the Structured Total Least
Norm approach adopted for multivariate polynomial approximate GCD
\cite{KYZ06}. A possible direction would be to look at the differential
subresultant sequence and the singular values of their inflated block matrices
\cite{EmiGalLom97}.

The differential polynomials defined in this paper are special case of more
general Ore polynomials, which have broader application in the solution of
differential and difference equations. In particular, we could potentially apply
our methods in the context of $q$-differentiation (Jackson differentiation) or
derivations on exponential polynomials. Ultimately, any Ore structure will have
a well-defined Sylvester-like matrix (see, e.g., \cite{GieKim13}).  However, the
numerical properties of different derivations may well be quite difficult or
even problematic, and may well introduce poles or other significant sources of
numerical instability.

We also hope, the results of this paper are a foundation for extending the
approximate polynomial toolbox to other problems with differential polynomials
and more general linear differential operators.  Much like approximate GCD, the
approximate GCRD is both a stepping stone and a key tool towards operations like
approximate factorization and (functional) solution of differential polynomials.
More immediately, computation of an approximate GCRD enables computation of a
corresponding approximate LCLM, and multiple GCRD's, and to multiple
differential variables (i.e., iterated Ore polynomials), which provide an
effective method for dealing with linear PDEs.


\section*{Acknowledgements}
The authors would like to thank George Labahn for
his comments.  The authors would also like to thank the two anonymous referees
for their careful reading and comments.

\bibliographystyle{plainnat}


\newcommand{\Gathen}{\relax}


\end{document}